\documentclass[11pt]{article}
\usepackage[a4paper, total={6in, 8in}]{geometry}

\usepackage{enumerate}
\usepackage{amsmath}
\usepackage{amssymb,latexsym}
\usepackage{amsthm}
\usepackage{color}
\usepackage{cancel}
\usepackage{graphicx}
\usepackage{cite}
\usepackage{changes}
\usepackage{hyperref}
\usepackage{comment}
\usepackage{amscd}
\usepackage{mathtools}

\newtheorem{theorem}{Theorem}[section]

\newtheorem{lemma}[theorem]{Lemma}
\newtheorem{corollary}[theorem]{Corollary}
\newtheorem{definition}[theorem]{Definition}
\newtheorem{proposition}[theorem]{Proposition}

\newtheorem{remark}[theorem]{Remark}

\DeclareMathOperator{\C}{\mathcal{C}}

\newcommand{\fqn}{\mathbb{F}_{q^n}}
\newcommand{\fq}{\mathbb{F}_{q}}

\newcommand{\fqt}{\mathbb{F}_{q^t}}

\newcommand{\cC}{{\mathcal C}}

\newcommand{\F}{{\mathbb F}}
\newcommand{\LL}{{\mathbb L}}
\newcommand{\KK}{{\mathbb K}}

\newcommand{\lmb}{\lambda}
\newcommand{\la}{\langle}
\newcommand{\ra}{\rangle}

\begin{document}

\title{Full weight spectrum one-orbit cyclic subspace codes}

\author{Chiara Castello, Olga Polverino\thanks{Corresponding author.}\,\,  and Ferdinando Zullo}

\maketitle   

\begin{abstract}
For a linear Hamming metric code of length $n$ over a finite field, the number of distinct weights of its codewords is at most $n$. 
The codes achieving the equality in the above bound were called full weight spectrum codes.
In this paper we will focus on the analogous class of codes within the framework of cyclic subspace codes. 
Cyclic subspace codes have garnered significant attention, particularly for their applications in random network coding to correct errors and erasures. 
We investigate one-orbit cyclic subspace codes that are \emph{full weight spectrum} in this context. Utilizing number theoretical results and combinatorial arguments, we provide a complete classification of full weight spectrum one-orbit cyclic subspace codes.
\end{abstract}

\section{Introduction}

Let $k$ be a non-negative integer with $k \leq n$, the set of all $k$-dimensional $\F_q$-subspaces of $\F_{q^n}$, viewed as an $\F_{q}$-vector space, forms a \textbf{Grassmannian space} over $\F_q$, which is denoted by $\mathcal{G}_{q}(n,k)\subseteq \mathcal{P}_{q}(n)$, where $\mathcal{P}_{q}(n)$ is the set of all the $\fq$-subspaces of $\fqn$. A \textbf{constant dimension subspace code} is a subset $\mathcal{C}$ of $\mathcal{G}_{q}(n,k)$ endowed with the metric defined as follows \[d(U,V)=2k-2\dim_{\F_q}(U \cap V),\]
where $U,V \in \mathcal{C}$. This metric is also known as \textbf{subspace metric}.
As usual, we define the \textbf{minimum distance} of $\mathcal{C}$ as
\[ d(\mathcal{C})=\min\{ d(U,V) \colon U,V \in \mathcal{C}, U\ne V \}. \]
Subspace codes have been recently used for the error correction in random
network coding, see \cite{KoetterK}. 
The first class of subspace codes studied was the one introduced in \cite{Etzion}, which is known as \textbf{cyclic subspace codes}.
A subspace code $\mathcal{C} \subseteq \mathcal{G}_q(n,k)$ is said to be \textbf{cyclic} if for every $\alpha \in \F_{q^n}^*$ and every $V \in \mathcal{C}$ then $\alpha V \in \mathcal{C}$.
If $\mathcal{C}$ coincides with $\mathrm{Orb}(S)$, for some subspace $S$ of $\fqn$, we say that $\mathcal{C}$ is a \textbf{one-orbit} cyclic subspace code and $S$ is said to be an its \textbf{representative}.

Similarly to the classical Hamming case, in \cite{heideweight} the following definition is given.

\begin{definition}\cite[Definition 2.4]{heideweight}\label{dist/weidistri}
Let $S\in\mathcal{G}_q(n,k)$ and let $\cC=\mathrm{Orb}(S)$. Define \[\omega_{2i}(\cC)=\mid \lbrace \alpha S\in\mathrm{Orb}(S)\colon \alpha\in\fqn^*, d(S,\alpha S)=2i\rbrace\mid\] for $i\in\{1,\dots,k\}$. We call $(\omega_2(\cC),\dots,\omega_{2k}(\cC))$ \textbf{distance/weight distribution} of $\cC$.
\end{definition}

Gluesing-Luerssen and Lehmann in \cite{heideweight} studied the weight distribution of cyclic subspace codes. In particular, they highlighted how the weight distribution may be used as a tool for a finer classification of cyclic orbit codes. Indeed, since 
$d(\alpha S,\beta S)=d(S,\alpha^{-1}\beta S)$ for any $\alpha,\beta \in \fqn^*$, the weight distribution provides the number of codeword pairs at a given distance.

In the classical framework of Hamming-metric codes, Delsarte in \cite{Del4par} studied the number of distinct distances for a code $\mathcal{C}$. When the code is linear this corresponds to investigate the distinct weights of the codewords in $\C$. 
Denote by $L(k,q)$ the maximum number of distinct non-zero weights a linear code of dimension $k$ over $\fq$ may have.
In \cite{Shi}, the authors proved that 
\[ L(k,q)\leq \frac{q^k-1}{q-1}. \]
As proved in \cite{Shi} and \cite{MWS}, codes with the equality in the above bound exist for any $q$ and $k$ and the codes with such a number of weights are called \textbf{maximum weight spectrum} codes. 
If we also fix the length of the code and we denote by  $L(n,k,q)$ the maximum number of distinct non-zero weights a linear code of dimension $k$ and length $n$ over $\fq$ may have. Again in \cite{Shi} it has been proved that
\[ L(n,k,q)\leq n, \]
and the parameters of the codes reaching the equality have been completely determined in \cite{1Ald}. Such codes have been named \textbf{Full Weight Spectrum} codes.
In this paper we investigate the subspace analog of full weight spectrum codes in the context of cyclic subspace codes.

As before, denote by $\mathcal{L}(n,k,q)$ the maximum number of distinct non-zero weights (in the sense of Definition \ref{dist/weidistri}) that a one-orbit cyclic subspace code can have. Note that such a number corresponds to the maximum number of possible distinct distances that a one-orbit cyclic subspace code can have. For a code $\mathcal{C}=\mathrm{Orb}(S)$ of $\mathcal{G}_q(n,k)$, the possible weights are determined by the intersection of $S$ with its cyclic shifts $\alpha S$. Indeed, we have that
\[ 2 \leq d(S,\alpha S)\leq 2k, \]
for any $\alpha \in \fqn$ such that $S\ne \alpha S$, and since such a value is always even, then 
\[ \mathcal{L}(n,k,q)\leq k. \]
Similarly to the Hamming metric case, the one-orbit cyclic subspace codes having exactly $k$ non-zero distinct weights are named \textbf{full weight spectrum (FWS)} codes.

In this paper, first we will show a classification result for one-orbit cyclic  subspace codes with minimum distance equals to two, for which we will be able to give information on the weight distribution.
These results will yield a classification of FWS one-orbit cyclic subspace codes, precisely we will prove the following main theorem.

\begin{theorem}
\label{thm:maintheorem}
    Let $\C$ be a one-orbit cyclic orbit code in $\mathcal{G}_q(n,k)$. Then $\C$ is a full weight spectrum code if and only if $\C=\mathrm{Orb}(S)$, where $S$ 
    is one of the following
    \begin{itemize}
        \item [(1)]$S=\langle 1,\lmb,\ldots,\lmb^{k-1}\rangle_{\fq}$ for some $\lmb \in \fqn \setminus \fq$, where
        \[ k\leq \begin{cases}
            \frac{[\fq(\lambda)\colon\fq]+1}{2} & \mbox{if} \,\,\, \dim_{\fq}(\fq(\lambda))<n,\\
            \frac{n}{2} & \mbox{if} \,\,\, \dim_{\fq}(\fq(\lmb))=n,
        \end{cases} \]
        \item [(2)] $S=\langle 1,\lmb,\ldots,\lmb^{l-1}\rangle_{\F_{q^2}}\oplus \lmb^l\fq$ for some $\lmb\in\fqn \setminus \F_{q^2}$, where $k=2l+1$, $n$ is even  and 
        $l< \frac{[\F_{q^2}(\lambda)\colon \F_{q^2}]}2$.
    \end{itemize}
\end{theorem}

The proof of this result relies on some steps. First, using the classification of critical pairs of the Cauchy-Davenport inequality for extension fields, we obtain a classification result for one-orbit cyclic subspace codes having minimum distance two. Since all FWS codes have all the possible weights,  the above classification result can be used as a starting point for the proof of Theorem \ref{thm:maintheorem}. 
Indeed, the one-orbit cyclic subspace codes having minimum distance two can be divided into two families of codes: for the first one, using some geometric arguments, we can characterize those codes that are FWS codes. For the latter family we need a deep analysis of their intersection properties and we prove that, apart from the family (2), all the other codes of this family have at least one zero in their weight distribution, and so they cannot be FWS codes.

\medskip

\textbf{Structure of the paper.} In Section \ref{sec:Equiv} we recall the notion of equivalence among cyclic subspace codes and we study some invariants.
Section \ref{sec:auxresults} is dedicated to some technical results and to recall properties of the dual in extension fields. In Section \ref{sec:startFWS}, we start the study of FWS cyclic subspace codes, we first point out some bounds on the parameters of an FWS code and then we prove that the families (1) and (2) described in Theorem \ref{thm:maintheorem} are FWS codes. For the Family (1) we completely determine the weight distribution.
Section \ref{sec:dist2} is devoted to the classification of one-orbit cyclic subspace codes with minimum distance two as a consequence of the classification of critical pairs of the linear analog of the Cauchy-Davenport inequality. In Section \ref{sec:proofnoFWScode} we provide the proof of Theorem \ref{thm:maintheorem}. To this aim, we study the weight distribution of the codes having minimum distance two and we show that, apart from those in (2) of Theorem \ref{thm:maintheorem}, there is always a zero in the weight distribution. The last section provides some consequences of our results for linear sets and rank-metric codes.

\section{Equivalence and invariants}\label{sec:Equiv}

The study of the equivalence for subspace codes was initiated by Trautmann in \cite{Trautmann2} and the case of cyclic subspace codes has been investigated in \cite{Heideequiv} by Gluesing-Luerssen and Lehmann. In particular, in the latter the authors focused on the case of cyclic orbit codes, which are by definition orbits of an $\fq$-subspace of $\fq^n$ under the action of a Singer subgroup $G$ of $\mathrm{GL}(n,q)$, that is a cyclic subgroup of order $q^n-1$ of the group of invertible matrices of order $n$ over $\fq$. 
By identifying $\fq^n$ with the field extension $\fqn$ as an $\fq$-vector space, we can identify the action of $\mathrm{GL}(n,q)$ on the set of $\fq$-subspaces of $\fq^n$ with the action of $\mathrm{GL}_n(q)$ on the set of $\fq$-subspaces of $\fqn$, where $\mathrm{GL}_n(q)$ is the group of $\fq$-vector space automorphisms of $\fqn$. In this way, the subgroup of $\mathrm{GL}_n(q)$ given by the multiplicative maps $x\mapsto ax$ for any $a\in\fqn^*$ is isomorphic to $\fqn^*$ and so it turns out to be a Singer subgroup of $\mathrm{GL}_n(q)$. In this case, the orbits take the following explicit form\[
\mathrm{Orb}_{\fqn^*}(S)=\lbrace \alpha S \mid \alpha\in\fqn^*\rbrace,
\]
where $S$ is an $\fq$-subspace of $\fqn$ and $\alpha S=\{\alpha \cdot s \,:\, s\in S\}$ is said to be a {\bf cyclic shift} of $S$. 
Since any Singer subgroup of $\mathrm{GL}_n(q)$ is conjugate to $\fqn^*$ (see \cite[II.7, pp.187]{singerconjugate}), we will only consider orbit codes under the action of $G=\fqn^*$ and we will denote by $\mathrm{Orb}(S):=\mathrm{Orb}_{\fqn^*}(S)$. Note that the normalizer $\mathrm{N}_{\mathrm{GL}_n(q)}(\fqn^*)$ of $\fqn^*$ in $\mathrm{GL}_n(q)$ is isomorphic to $\mathrm{Gal}(\fqn |\fq)\rtimes \fqn^*$ (see \cite[Theorem 2.4]{Heideequiv}).   \\

By \cite[Definition 3.5 and Theorem 2.4]{Heideequiv}, we can give the following definitions.

\begin{definition}
\label{def:frobiso}
Let $\mathcal{C}_1, \mathcal{C}_2\subseteq\mathcal{G}_q(n,k)$. Then $\mathcal{C}_1$ and $\mathcal{C}_2$ are called \textbf{(linearly) isometric} if there exists an isomorphism $\psi\in\mathrm{GL}_n(q)$ such that $\psi(\mathcal{C}_1)=\mathcal{C}_2$, where
$$\psi(\mathcal{C}_1)=\lbrace \psi(V)\colon V\in\mathcal{C}_1\rbrace.$$ In this case $\psi$ is called a \textbf{(linear) isometry} between $\mathcal{C}_1$ and $\mathcal{C}_2$. In the special case, where $\mathcal{C}_1=\mathrm{Orb}(S_1)$, $\mathcal{C}_2=\mathrm{Orb}(S_2)$ and $\psi(\mathcal{C}_1)=\mathcal{C}_2$ for some $\psi\in \mathrm{N}_{\mathrm{GL}_n(q)}(\fqn^*)$, we call the cyclic orbit codes $\mathrm{Orb}(S_1)$ and $\mathrm{Orb}(S_2)$ \textbf{Frobenius-isometric} and $\psi$ a \textbf{Frobenius isometry}. Also, if $\mathcal{C} \subseteq\mathcal{G}_q(n,k)$, the \textbf{automorphism group} of $\cC$ is the group of linear isometries that fix  $\cC$, that is $Aut(\cC):=\{\psi\in\mathrm{GL}_n(q) \, : \; \psi(\cC)=\cC\}$.
\end{definition}

In the next we will introduce some integers associated with an $\fq$-subspace of $\fqn$ that are invariant under the action of  $\mathrm{N}_{\mathrm{GL}_n(q)}(\fqn^*)$.
\begin{definition} \label{def:parameters}
For any divisor $t$ of $n$ and for any $\fq$-subspace $S$ of $\fqn$, we denote by 
\[
\delta_t(S):=\dim_{\fqt}(\langle S\rangle_{\F_{q^t}}),
\]
\[
h_t(S):=\max\lbrace \dim_{\fqt}(\Bar{S})\colon \Bar{S}\text{ is an }\fq\text{-subspace of } S \text{ and }\Bar{S} \text{ is an } \fqt\text{-subspace of } \fqn\rbrace
\]
and
\[
m(S)=\min\lbrace t\colon t\,| \, n \, \,\mbox{and} \, \, \delta_t(S)=1   \rbrace.
\]
Moreover, if $m=m(S)$, then $\F_{q^{m}}$ is the smallest subfield of $\fqn$ containing  a cyclic shift of $S$. If $m(S)=n$, then we will say that the subspace $S$ is   \textbf{generic}  in $\fqn$.
\end{definition}

Gluesing-Luerseen and Lehmann in \cite{Heideequiv} study the automorphisms group of one-orbit cyclic subspace codes and completely classify one-orbit cyclic subspace codes generated by generic subspaces. 

\begin{theorem}\cite[Theorem 6.2(a)]{Heideequiv}\label{thm:classisomlinear}
Let $\cC=\mathrm{Orb}(S)$ and $\cC'=\mathrm{Orb}(S')$ be  one-orbit cyclic subspace codes in $\mathcal{G}_q(n,k)$ such that $S'$ is generic. Then $\cC$ and $\cC'$ are linearly isometric if and only if they are Frobenius isometric.
\end{theorem}

Clearly, $\delta_t(\cdot)$, $h_t(\cdot)$ and $m(\cdot)$ are invariant under the action of  $\fqn^*$ 
and  it makes sense to give the following definition (see \cite[Definition 4.5]{Heideequiv}).

\begin{definition} \label{def: parmeters for code}
Let $\cC$ be a one-orbit cyclic subspace code in $\mathcal{G}_q(n,k)$ and let $S\in \cC$. For any divisor $t$ divisor of $n$, let denote 
\[
\delta_t(\cC):=\delta_t(S), \,\,\, h_t(\cC):=h_t(S) \,\, \, \mbox{and} \, \, \, m(\cC):=m(S).
\]

\end{definition}

\begin{remark}\label{rem:invariants}
Note that the integers  $\delta_t(\cdot)$, $h_t(\cdot)$ and $m(\cdot)$ are invariant also under the action of $\mathrm{Gal}(\fqn |\fq)$, hence by Theorem \ref{thm:classisomlinear}, if $\cC_1$ and $\cC_2$ are linearly isometric one-orbit cyclic subspaces codes defined by generic subspaces, then 
\[\delta_t(\cC_1)=\delta_t(\cC_2), \,\, h_t(\cC_1)=h_t(\cC_2),\,\, m(\cC_1)=m(\cC_2) \]
for any divisor $t$ of $n$.
 
\end{remark}
 In the next corollary, as a consequence of \cite[Theorems 2.4 and 4.4]{Heideequiv}, we prove that the parameter $m(\cC)$ for the one-orbit code $\cC=\mathrm{Orb}(S)$ is invariant under the action of $\mathrm{GL}_n(q)$ without any assumption on the subspace $S$. 

\begin{corollary}
  Let $\cC_1$ and $\cC_2$ be one-orbit cyclic subspace codes in $\mathcal{G}_q(n,k)$. If $\cC_1$ and $\cC_2$ are linearly isometric then $m(\cC_1)=m(\cC_2)$. 
\end{corollary}
\begin{proof}
    Let $m_i:=m(\cC_i)$ for $i=1,2$. By \cite[Theorem 4.4]{Heideequiv}
    \begin{equation}\label{eq:subgroups}
        \mathrm{GL}_{n/m_i}(q^{m_i}) \unlhd \mathrm{Aut}(\cC_i) \leq \mathrm{N}_{\mathrm{GL}_{n}(q)}(\mathrm{GL}_{n/m_i}(q^{m_i})),
    \end{equation}
    for $i\in \{1,2\}$, where $\mathrm{GL}_{n/m_i}(q^{m_i})=\{\psi \in \mathrm{GL}_{n}(q) \,:\, \psi \, \mbox{is} \,\,  \mathbb{F}_{q^{m_i}} \,\, \mbox{-linear}\}$. 
    Since by \cite[Theorem 2.4(d)]{Heideequiv}
    $$\mathrm{N}_{\mathrm{GL}_{n}(q)}(\mathrm{GL}_{n/m_i}(q^{m_i})) \cong  \mathrm{Gal}(\mathbb{F}_{q^{m_i}} | \fq) \rtimes  \mathrm{GL}_{n/m_i}(q^{m_i}),$$
    from \eqref{eq:subgroups} we get 
    \begin{equation} \label{eq: order aut group}
        |\mathrm{GL}_{n/m_i}(q^{m_i})| \leq  |\mathrm{Aut}(\cC_i)| \leq m_i |\mathrm{GL}_{n/m_i}(q^{m_i})|,
 \end{equation}
    for $i\in \{1,2\}$. If $\cC_1$ and $\cC_2$ are linearly isometric, then $\mathrm{Aut}(\cC_1)$ and $\mathrm{Aut}(\cC_2)$ are isomorphic groups and hence they have the same size. Now, by contradiction assume that $m_1\neq m_2$, let $m_2>m_1$ and  $m_i=n/t_i$ for $i\in\{1,2\}$, from (\ref{eq: order aut group})   we have
  \begin{equation} \label{eq: order aut group2}
 \prod_{i=0}^{t_1-1} (q^n-q^{m_1i})       \leq  |\mathrm{Aut}(\cC_1)|=|\mathrm{Aut}(\cC_2)| \leq m_2\prod_{i=0}^{t_2-1} (q^n-q^{m_2i}), 
  \end{equation}
  and hence, since $q^n-q^{m_1i}>q^{n-1}$ for any $i \in \{0,\ldots,t-1\}$ and $m_2\leq n\leq q^n$ for any positive integer $n$, we get  
\begin{equation} \label{eq: order aut group3}
 q^{(n-1)t_1}       <  |\mathrm{Aut}(\cC_1)|=|\mathrm{Aut}(\cC_2)| < q^{n(t_2+1)}. 
  \end{equation}
By (\ref{eq: order aut group3}) it follows that $t_1=t_2+1$ and by \eqref{eq: order aut group2}, since $m_2>m_1$, we easily get
    \begin{equation} \label{eq: order aut group4}
 q^n-q^{t_2m_1}    < m_2,  
  \end{equation}
  and, since $m_2 <q^{m_2}$, hence
  \begin{equation*} 
 q^{(t_2-1)m_2}    \leq q^{t_2m_1-m_2}    
  \end{equation*}
   which gives $m_2\leq m_1$, a contradiction. So $m_1=m_2$.
\end{proof}

\section{Auxiliary results}\label{sec:auxresults}

In order to prove Theorem \ref{thm:maintheorem}, pur main result, we need to recall some definitions and properties, and to prove some preliminary results that will be used later.\\
As first, we recall that, given an $\fq$-subspace $S$ of $\fqn$, the {\bf stabilizer of $S$} is 
\[
H(S)=\lbrace x\in \fqn^*\colon xS= S\rbrace \cup \{0\}.
\]
Note that $H(S)$ is a subfield of $\fqn$ and $S$ is linear over $H(S)$. This implies that also $S\cap\alpha S$ is linear over $H(S)$ for any $\alpha\in\fqn$. The following lemma gives us information on the weight distribution of a one-orbit cyclic subspace code in relation with the stabilizer of one of its representatives.

\begin{lemma}
\label{lem:w2i>0}
Let $\C=\mathrm{Orb}(S)\subseteq \mathcal{G}_q(n,k)$ and let $(\omega_2(\C),\dots,\omega_{2k}(\C))$ be its weight distribution. If $\omega_{2i}(\C)>0 $ for some $i\in\{1,\dots,k\}$, then $k\equiv i \pmod{[H(S)\colon \fq]}$.
\end{lemma}
\begin{proof}
If $\omega_{2i}(\C)>0$ then there exists $\alpha\in\fqn\setminus\fq$ such that $d(S,\alpha S)=2k-2\cdot \dim_{\fq}(S\cap\alpha S)=2i$, i.e. $\dim_{\fq}(S\cap\alpha S)=k-i$. Since $S\cap\alpha S$ is an $H(S)$-subspace of $\fqn$, we have the desired congruence.
\end{proof}

To prove our main result, we will investigate some properties on the weight distribution of the codes of the families given in Theorem \ref{thm:maintheorem}. To this aim,  we will need the orthogonal complement of a subspace defined by the trace function.

\begin{remark}\label{rem:dual}
The trace function $\mathrm{Tr}_{\fqn/\fq}\colon a\in\fqn \mapsto \displaystyle\sum_{i=0}^{n-1} a^{q^i}\in\fq$  
 is an  $\fq$-linear map of $\fqn$ and it also defines a nondegenerate, symmetric, bilinear form as follows:
\[
(a,b)\in\fqn\times\fqn \mapsto \mathrm{Tr}_{\fqn/\fq} 
(ab)\in \fq.
\]
Therefore, for any subset $S\subseteq\fqn $ we may define the orthogonal complement 
 (or dual subspace)  $S^{\perp}$ of $S$ in the following way:
\[
S^{\perp}=\lbrace a\in\fqn \colon \mathrm{Tr}_{\fqn/\fq}(ab)=0 \text{ for any b}\in S\rbrace.
\]
Note that if $\fqt$ is a subfield of $\fqn$ and $W$ is an $\fqt$-subspace of $\fqn$, then $W^{\perp}$ is an $\fqt$-subspace of $\fqn$ as well, i.e. the orthogonal complement $\perp$ preserves the  field of linearity of the  $\fq$-subspaces.  
\end{remark}

As we will see in the next section, the $\fq$-subspaces of $\fqn$ that admit a basis defined by consecutive powers of the same element,  define cyclic subspace codes  for which it is possible to determine the weight distribution. Also, in some cases there is no zero in the weight distribution. In what follows, we will say that an $\fq$-subspace $S$ of $\fqn$ of dimension $k$ admits a {\bf polynomial basis}, if there exist $\rho, \lambda  \in \fqn ^*$ such that $S=\rho \langle 1, \lambda,.....\lambda^{k-1} \rangle$. In \cite[Proposition 2.9]{NPSZminsize} it has been proved that the orthogonal complement of a subspace which has a polynomial basis, has a polynomial basis as well.

\begin{proposition}\cite[Proposition 2.9]{NPSZminsize}
\label{prop:dualbasisbasepol}
Let $\lambda\in\fqn$ such that $\fq(\lambda)=\fqn$. Let $f(x)=a_0+a_1x+\dots+a_{n-1}x^{n-1}+x^n$ be the minimal polynomial of $\lambda$ over $\fq$. Let $W=\langle 1, \lambda, \dots, \lambda^{l-1}\rangle_{\fq}$ with $l\in\lbrace 1, \dots, n-1\rbrace$. Then $W^{\perp}=\delta^{-1}\langle 1, \lambda, \dots, \lambda^{n-l-1}\rangle_{\fq}$, where $\delta=f'(\lambda)$.
\end{proposition}

The following proposition gives some properties of $\fq$-subspaces of $\fqn$ contained in an $\fqt$-subspace of $\fqn$.

\begin{proposition}
    \label{prop:hyperplane}
    Let $t\mid n$ and let $U$ be an $\fq$-subspace of $\fqn$ contained in an $\fqt$-subspace $W$ of $\fqn$ of dimension $l$. Suppose that $\dim_{\fq}(U)=lt-1=\dim_{\fq}(W)-1$. Then there exists an $\fqt$-subspace $T$ of $\fqn$ such that 
    \[
    \dim_{\fqt}(T)=l-1 \,\,\text{ and } \, \,T\subseteq U.
    \]
    Also, if $\xi\fqt\subseteq W$, then either $\xi\fqt\subseteq U$ or $\dim_{\fq}(U\cap\xi\fqt)= t-1$.
\end{proposition}
\begin{proof}
    Let $n=st$ and note that $\dim_{\fq}(U^{\perp})=(s-l)t+1$ and $\dim_{\fq}(W^{\perp})=(s-l)t$. Since $U\subseteq W$, it results $W^{\perp}\subseteq U^{\perp}$ i.e.
    \[
    U^{\perp}=W^{\perp} \oplus \langle v\rangle_{\fq} \subset W^{\perp} \oplus \langle v\rangle_{\fqt} 
    \]
    for some $v\in U^{\perp}\setminus W^{\perp}$. Let $V=W^{\perp} \oplus \langle v\rangle_{\fqt}$,  then by Remark \ref{rem:dual}, $V$ is an $\fqt$-subspace of dimension $s-l+1$ over $\fqt$ and hence $V^{\perp}=T$ is an $\fqt$-subspace of dimension $l-1$ and  clearly, $T\subseteq U$. Finally, the last part of the assertion follows from the fact that $U$ is an hyperplane of $W$ seen as an $\fq$-subspace of $\fqn$.
\end{proof}

\section{Existence and constructions of FWS codes}\label{sec:startFWS}

This section analyzes the existence of FWS codes and we will deal with two families of such codes. First we will derive some necessary conditions on the parameters of an FWS code, then we will present a first family of FWS codes for each possible choice of  admissible parameters. Finally, we will present another non-equivalent family of FWS codes, which relies on some algebraic properties of the extension of $\fq$-subspaces of $\fqn$  over the quadratic extension $\F_{q^2}$ when $n$ is even.

\subsection{Bounds on the parameters of an FWS code}

In this section, we will see that not all potential parameter values permit the existence of 
 an FWS code.  In fact, the following proposition will establish necessary conditions for the existence of such codes in $\mathcal{G}_q(n,k)$. These conditions will involve the parameters $n$, $k$, and $m({\mathcal C})$ of the code.

\begin{proposition}
Let $\cC=\mathrm{Orb}(S)$ with $S\in\mathcal{G}_q(n,k)$ and let $m:=m(\cC)$. If $m<n$ and $\omega_{2k-2}(\C)\neq 0$ then $k\leq \frac{m+1}{2}$; if $m=n$ and $\omega_{2k}(\C)\neq 0$ then $k\leq n/2$. 
\end{proposition}
\begin{proof}
Let $m<n$ and $\omega_{2k-2}(\C)\neq 0$. Note that, up to a cyclic shift, we may suppose that $S$ is an $\F_{q}$-subspace of the subfield $\F_{q^m}$ of $\fqn$. Since $\omega_{2k-2}(\C)\neq 0$, then  
there exists $\alpha\in\fqn^*$ such that $d(S,\alpha S)=2k-2$. This means that $\dim_{\fq}(S \cap \alpha S)=1$ and so
 there exist $s,s'\in S\setminus\lbrace 0 \rbrace$ such that $s=\alpha s'$, and hence 
$\alpha=ss'^{-1}\in\F_{q^m}.$
Thus we have
\[
m \geq \dim_{\fq}(S+\alpha S)=2k-\dim_{\fq}(S\cap \alpha S)=2k-1
\]
and so $k\leq \frac{m+1}{2}$.
Let $m=n$ and  $\omega_{2k}(\C)\neq 0$. Then there exists $\alpha \in \fqn$ such that $d(S,\alpha S)=2k$, i.e.
\[
n \geq \dim_{\fq}(S+\alpha S)=2k-\dim_{\fq}(S\cap \alpha S)=2k
\]
 and so
$k\leq n/2$.

\end{proof}

As an immediate consequence we get the following corollary.

\begin{corollary}\label{cor:necessaryconditionsFWS}
Let $\C=\mathrm{Orb}(S)\subseteq \mathcal{G}_q(n,k)$ be an FWS code and let $m=m(\C)$. If $m<n$, then $k\leq \frac{m+1}{2}$; if $m=n$, then $k\leq n/2$. 
\end{corollary}

In the next section we will show that these bounds are tight.

\subsection{A family of FWS code}

Now we exhibit some examples of FWS codes obtained by subspaces that have a polynomial basis, and so belonging to Family (1) of Theorem \ref{thm:maintheorem}. To this aim, we start from a technical lemma, which follows the arguments of \cite[Lemma 2.5]{JVdV}.\\
We will adopt the following notation: ${\fq}[x]_{\leq k-1}$ denotes the set of the polynomials in the variable $x$ with coefficients in $\fq$ and degree at most $k-1$.

\begin{lemma}
    \label{Lem:lemmamu}
Let $\fqn=\fq (\lambda)$ and let $S=\langle 1, \lmb, \dots, \lmb^{k-1}\rangle_{\fq}\in\mathcal{G}_q(n,k)$  with $1<k<n-1$. \\
If $k\leq n/2$, then for any $j\in\lbrace 1,\dots,k\rbrace$, 
\[
\dim_{\fq}(S\cap\mu S)=j \, \text{ if and only if } \, \mu=\frac{p(\lmb)}{q(\lmb)}
\]
with $p(x),q(x)\in\fq[x]_{\leq k-1} $ such that $\gcd(p(x),q(x))=1$ and $k-j=\max\lbrace  \deg(p(x)),\deg(q(x))\rbrace$.\\
In particular, 
\[
S\cap\mu S=\lbrace p(\lmb)a(\lmb)\colon a(x)\in\fq[x]_{\leq j-1} \rbrace=\mu \lbrace  q(\lmb)a(\lmb)\colon a(x)\in\fq[x]_{\leq j-1} \rbrace.
\]
If $n/2< k < n-1$, then for any $i\in\lbrace 2k+1-n,\dots, k\rbrace$ 
\[
\dim_{\fq}(S\cap\mu S)=i \text{ if and only if } \mu=\frac{p(\lmb)}{q(\lmb)}
\]
with $p(x),q(x)\in\fq[x]_{\leq k-1} $ such that $\gcd(p(x),q(x))=1$ and $k-i=\max\lbrace \deg(p(x)),\deg(q(x)\rbrace$.
\end{lemma}
\begin{proof}
Let $1<k\leq n/2$ and suppose $\mu=\frac{p(\lmb)}{q(\lmb)}$ where $p(x),q(x)\in\fq[x]_{\leq k-1} $ such that $\gcd(p(x),q(x))=1$ and $k-j=\max\lbrace  \deg(p(x)),\deg(q(x))\rbrace$ for $j\in\lbrace 1,\dots,k\rbrace$.\\
For any $a(x)\in\fq[x]_{\leq j-1}$ we have that 
\[
p(\lmb)a(\lmb)\in S \text{ and } q(\lmb)a(\lmb)\in S,
\]
thus
\[
H=\lbrace p(\lmb)a(\lmb)\colon a(\lmb)\in\fq[\lmb]_{\leq j-1} \rbrace=\mu \lbrace  q(\lmb)a(\lmb)\colon a(\lmb)\in\fq[\lmb]_{\leq j-1} \rbrace\subseteq S\cap\mu S.
\]
Using \cite[Lemma 2.5]{JVdV}, one can show that $H=S\cap \mu S$ and hence
$\dim_{\fq}(S\cap\mu S)=j$.

Conversely, suppose that $\dim_{\fq}(S\cap\mu S)=j$ for some $j\in\lbrace 1,\dots,k\rbrace$. Then there exist $p(\lmb),q(\lmb)\in S$ such that \[
\mu=\frac{p(\lmb)}{q(\lmb)}.
\]
Of course we may suppose $\gcd(p(x),q(x))=1$ and, since $k\leq n/2$, the pair of coprime polynomials $(p(x),q(x))$ such that $\mu=\frac{p(\lmb)}{q(\lmb)}$ is uniquely determined (up to scalars in $\fq$). Moreover, by the first implication, we get that $S\cap\mu S=H$ and  hence   $$k-j=\max\lbrace\deg(p(x)),\deg(q(x))\rbrace.$$
Let $n/2< k< n-1$, then by Proposition \ref{prop:dualbasisbasepol} we know that \[
S^{\perp}=\rho\langle 1, \lmb, \dots, \lmb^{n-k-1}\rangle_{\fq}
\]
for some $\rho\in\fqn^*$. In particular $1<n-k<n/2$ and hence we may apply the first case to $S^{\perp}$. Then we get that 
\[
\dim_{\fq}(S^{\perp}\cap\mu S^{\perp})=j \text{ if and only if }\mu=\frac{p(\lmb)}{q(\lmb)},
\]
with $p(x),q(x)\in\fq[x]_{\leq n-k-1} $ such that $\gcd(p(x),q(x))=1$ and 
\[ n-k-j=\max\lbrace  \deg(p(x)),\deg(q(x))\rbrace \mbox{ for } j\in\lbrace 1,\dots,n-k\rbrace.\]
Moreover, since $(\mu S)^{\perp}=\mu^{-1}S^\perp$,
\[
j=\dim_{\fq}(S^{\perp}\cap\mu S^{\perp})=\dim_{\fq}(S+\mu^{-1} S)^{\perp}=n-\dim_{\fq}(S+\mu^{-1} S)=n-2k+\dim_{\fq}(S\cap\mu^{-1} S).
\]
Thus, since $\dim_{\fq}(S\cap\mu^{-1} S)=\dim_{\fq}(S\cap\mu S) $, we have that for $j\in\lbrace 1,\dots, n-k\rbrace$, 
\[
\dim_{\fq}(S^{\perp}\cap\mu S^{\perp})=j \text{ if and only if }\dim_{\fq}(S\cap\mu S)= 2k+j-n.
\]
The statement follows setting $i=2k+j-n$.
\end{proof}

We can now use the above lemma to determine the weight distribution of one-orbit cyclic subspace codes defined by subspaces admitting a polynomial basis. The determination of the weight distribution is related to the problem of determining the weight distribution of certain type of linear sets on the projective line (see \cite[Theorem 2.7]{JVdV}) and to the problem of counting the number of coprime polynomials over finite fields (see \cite{countcoppol}).

\begin{theorem}\label{thm:polbasiscomplete}
    Let $\lambda \in \fqn\setminus\fq$ and let  $t:=[\fq(\lambda):\fq]$. Let $S=\la 1,\lambda,\ldots,\lambda^{k-1}\ra_{\fq}$ with $1< k<t$. Then for the code  $\cC=\mathrm{Orb}(S)$ we have the following.\\
    If $k\leq t/2$, then 
    \[ 
    \omega_{2i}(\C)= \begin{cases}
    (q+1)q^{2i-1} & \text{if}\,\, i \in \{1,\ldots,k-1\},\\ & \\
    \frac{q^n-q^{2k-1}}{q-1} & \text{if}\,\, i=k.
    \end{cases}
    \]
    If $k>t/2$, then 
     \[ 
    \omega_{2i}(\C)= \begin{cases}
    (q+1)q^{2i-1} & \text{if}\,\, i \in \{1,\ldots,t-k-1\},\\ 
     & \\
    \frac{q^t-q^{2(t-k)-1}}{q-1} & \text{if}\,\, i=t-k \\
    & \\
      0  & \text{if}\,\, i \in \{t-k+1,\ldots,k-1\} \text{ and } t-k+1\ne k, \\
        & \\
    \frac{q^n-q^{t}}{q-1} & \text{if}\,\, i=k.
    \end{cases}
    \]
\end{theorem}
\begin{proof}
    We start by noting that
    \[\omega_{2i}(\C)=\mid \lbrace \mu S\in \cC \colon \mu\in\fqn\setminus \fq \, \text{ and } \, \dim_{\fq}(S\cap \mu S)=k-i\rbrace\mid\] for any $i\in \{1,\ldots,k\}$.
    Suppose that $k\leq t/2$. By Lemma \ref{Lem:lemmamu}, $\dim_{\fq}(S\cap \mu S)=j$ if and only if $\mu=\frac{p(\lmb)}{q(\lmb)}$
    with $p(x),q(x)\in\fq[x]_{\leq k-1} $ such that $\gcd(p(x),q(x))=1$ and $k-j=\max\lbrace  \deg(p(x)),\deg(q(x))\rbrace$.
    Note that if $\mu \notin \fq(\lambda)=\F_{q^t}$ then $\dim_{\fq}(S\cap \mu S)=0$. Hence, assume that $\mu \in \F_{q^t}$.
    Therefore, we have a one-to-one correspondence between those $\mu \in \F_{q^t}$ such that $\dim_{\fq}(S \cap \mu S)=k-i$ for some $i\in \{1,\dots, k-1\}$ and pairs of coprime polynomials in $(\fq[x]_{\leq k-1})^2$ (up to scalars of $\fq$).
    In order to find the $\omega_{2i}$'s we need to determine the number of such a pairs for which the maximum degree among the two polynomials is $i$.
    These values have been computed in the proof of \cite[Lemma 2.6]{JVdV}, 
     obtaining that $\dim_{\fq}(S\cap \mu S)=k-i$ for $q^{2i+1}-q^{2i-1}$ different values of $\mu$ in $\F_{q^t}$ for any $i \in \{1,\ldots,k-1\}$ and $\dim_{\fq}(S\cap \mu S)=0$ for $q^{t}-q^{2k-1}$ different values of $\mu$ in $\F_{q^t}$, as this corresponds to those $\mu$ which cannot be written as the ratio of two polynomials over $\fq$ in $\lambda$ of degree less than $k$. Taking into account that for any $\mu \in \fqn\setminus \F_{q^t}$ we have $\dim_{\fq}(S\cap \mu S)=0$, the number of non-zero $\mu\in\fqn$ such that $d(S,\mu S)=2k$ is $(q^n-q^t)+(q^t-q^{2k-1})$.
    In particular, $\omega_{2(k-1)}(\C)>0$ and hence by Lemma \ref{lem:w2i>0} we get $H(S)=\fq$ and hence $\mu S=S$ if and only if $\alpha \in \fq^*$. So 
    $$\omega_{2i}(\C)=\frac{q^{2i+1}-q^{2i-1}}{q-1} \,\, \text{ for }i \in \{1,\ldots,k-1\} \, \text{ and }  \, \omega_{2k}(\C)=\frac{q^n-q^{2k-1}}{q-1}.$$\\
    Suppose now that $k>t/2$ and consider $S^\perp$ the dual of $S$ in $\F_{q^t}$ as in Remark \ref{rem:dual}.
    By Proposition \ref{prop:dualbasisbasepol}, we have that 
    \[S^\perp=\delta \la 1,\lambda,\ldots,\lambda^{t-k-1} \ra_{\fq},\]
    for some $\delta \in \F_{q^t}^*$. Since $\dim_{\fq} (S^{\perp})= t-k < t/2$ and \[ \dim_{\fq}(S\cap \mu S)=2k-t+\dim_{\fq}(S^\perp\cap \mu S^\perp), \] for any $\mu \in \fqt$, we get 
   \[ d(S, \mu S)=d(S^\perp, \mu S^\perp) \quad \forall \mu \in \fqt^*, \]
   so for any $\mu \in \fqt^*$, we can use the above arguments to get the second part of the assert. In particular, there exist $q^t-q^{2(t-k)-1}$ elements $\mu \in \F_{q^t}$  such that $d(S, \mu S)=2(t-k)$, whereas there are no elements $\mu$ in $\fqt$ such that $d(S, \mu S)=2i$ for $i\in \{t-k+1,\dots,k-1\}$ if $t-k+1<k$. Finally, if $\mu \in \fqn \setminus \fqt$ we have $\dim_{\fq} (S\cap \mu S)=0$ and hence $d(S, \mu S)=2k$.
    
\end{proof}


\begin{remark}
    Note that the weight distribution depends on the value $\lambda$ only for the degree extension $\fq(\lambda)/\fq$.
\end{remark}

As a consequence we can determine under which conditions on the parameters a subspace with a polynomial basis gives an FWS code.

\begin{corollary}
\label{cor:existFWSpolbas}
    For any $\lambda \in \fqn\setminus\fq$, let $t=\dim_{\fq}(\fq(\lambda))$ and $S=\la 1,\lambda,\ldots,\lambda^{k-1}\ra_{\fq}$ with $1< k<t$. Then $\cC=\mathrm{Orb}(S)$ is an FWS code in $\mathcal{G}_q(n,k)$ in the following cases:
    \begin{enumerate}
        \item $t<n$ and $k\leq \frac{t+1}{2}$;
        \item   $t=n$ and $k\leq \frac{n}{2}$. 
    \end{enumerate}
    Also, in both cases $m(\cC)=t$ and the weight distribution of $\cC$ is 
    \[\omega(\C)=\left((q+1)q,(q+1)q^{3},\dots,(q+1)q^{2i-1},\dots, (q+1)q^{2(k-1)-1}, \frac{q^n-q^{2k-1}}{q-1}\right)
    \]
    if $k\leq \frac{t}{2}$; whereas the weight distribution of $\cC$ is
     \[\omega(\C)=\left((q+1)q,(q+1)q^{3},\dots,(q+1)q^{2i-1},\dots,, (q+1)q^{2(t-k-1)-1},\frac{q^t-q^{2k-3}}{q-1}, \frac{q^n-q^{2k-1}}{q-1}\right)
    \]
     if $k= \frac{t+1}{2}$ and $t<n$. 
\end{corollary}
\begin{proof}
If $k\leq t/2$, then 
by Theorem \ref{thm:polbasiscomplete} we have $\omega_{2i}(\C)>0$ for any $i\in\lbrace 1, \dots, k\rbrace$, so $\cC$ has exactly $k$ non-zero distinct weights, i.e. $\cC$ is an FWS code.
If $k=\frac{t+1}{2}$ and $t<n$, then by Theorem \ref{thm:polbasiscomplete} we have that $\omega_{2j}(\C)>0$ for any $j\in\lbrace 1, \dots, \frac{t-3}{2}\rbrace$. Also, $\omega_{t-1}(\C)=\frac{q^t-q^{t-2}}{q-1}>0$ and $\omega_{t+1}(\C)=\frac{q^n-q^t}{q-1}>0$. Therefore $\cC$ has $k$ non-zero distinct weights, i.e. $\cC$ is an FWS code. Finally,  since $\fqt=\fq(\lambda)=\fq(S)$, we have that $m(\cC)=t$ and the weight distributions comes from Theorem \ref{thm:polbasiscomplete}. 
\end{proof}

Therefore, for any admissible values of $k,t,n$ and $q$ (that is, satisfying the bound in Corollary \ref{cor:necessaryconditionsFWS}) there exists an FWS  one-orbit cyclic subspace code $\cC$  in $\mathcal{G}_q(n,k)$. Precisely, we get the following result.

\begin{corollary}
    For any positive integer $n$, for any divisor $t$ of $n$, with either $1<t<n$ and $k\leq (t+1)/2$ or $t=n$ and $k\leq n/2$, and any prime power $q$, there exists $\C=\mathrm{Orb}(S)$ in $\mathcal{G}_{q}(n,k)$ with $m(\C)=t$ which is an FWS code.
\end{corollary}
\begin{proof}
    
    The proof follows by Corollaries \ref{cor:necessaryconditionsFWS} and \ref{cor:existFWSpolbas}.
    
\end{proof}

\subsection{Other examples of FWS codes}

In this section, another family of FWS codes will be described. We will see that these examples arise from subspaces belonging to Family (2) of Theorem \ref{thm:maintheorem}. To achieve this goal, we need the following proposition, that will allow us to get some bounds on the parameters of these codes, and then a technical proposition to get information on the weight distribution of such codes.

\begin{proposition}
    \label{prop1 h'h}
Let $n$ be an even positive integer and let 
\begin{equation} \label{EqBarS}
    S=\Bar{S}\oplus b \fq,
\end{equation}
where $\Bar{S}$ is an $\F_{q^2}$-subspace of $\fqn$ of dimension $l$ over $\F_{q^2}$ and $b\in \fqn \setminus \Bar{S}$. Denoted by $Y=\langle S\rangle_{\F_{q^2}}=\Bar{S}\oplus b \F_{q^2}$, for any $\mu \in \fqn$ one of the following occurs: 
\begin{itemize}
    \item [I)] $\Bar{S}\cap\mu\Bar{S}=S\cap\mu S=Y\cap\mu Y$ and $S\cap\mu S$ is an $\F_{q^2}$-subspace;
    \item[II)] $\dim_{\fq}(\Bar{S}\cap\mu\Bar{S})=\dim_{\fq}(Y\cap\mu Y)-2$, $\dim_{\fq}(S\cap\mu S) \in \{\dim_{\fq}(\Bar{S}\cap\mu\Bar{S}),\dim_{\fq}(\Bar{S}\cap\mu\Bar{S})+1\}$ and $S\cap\mu S$ is an $\F_{q^2}$-subspace if and only if $S\cap \mu S=\bar{S}\cap \mu\bar{S}$;
    \item[III)]  $\dim_{\fq}( S\cap\mu S)=2+\dim_{\fq}(\Bar{S}\cap\mu \Bar{S})=\dim_{\fq}( Y\cap\mu Y)-2$ and $S\cap\mu S$ is not an $\F_{q^2}$-subspace of $\fqn$.
\end{itemize}
In particular, if $\dim_{\fq}( S\cap\mu S)=2l$, then either $\mu \in H(Y)$ or $\mu \in H(\Bar{S})$.
\end{proposition}
\begin{proof}
Note that $\Bar{S}\subset S\subset Y$, then we easily get
\begin{equation}
    \label{eq:inter}
    \Bar{S}\cap\mu\Bar{S}\subseteq S\cap\mu S\subseteq Y\cap\mu Y
\end{equation}
and 
\begin{equation}
\label{eq:sum}
\Bar{S}+\mu\bar{S}\subseteq S+\mu S\subseteq Y+\mu Y.  
\end{equation}
Moreover, $\Bar{S}\cap\mu\Bar{S}$ and $Y\cap\mu Y$ are $\F_{q^2}$-subspaces of $\fqn$. Let $\dim_{\fq}(\bar{S}\cap\mu\Bar{S})=2h'$ and $\dim_{\fq}(Y\cap\mu Y)=2h$. Then by  \eqref{eq:inter} we have $h'\leq h$. Moreover, using the Grassmann formula we get $\dim_{\fq}(S+\mu S)=4l+2-\dim_{\fq}(S\cap\mu S)$ and by  \eqref{eq:sum} we obtain that
\begin{equation}
    4l-\dim_{\fq}(\Bar{S}\cap\mu\Bar{S}) \leq \dim_{\fq}(S+\mu S)\leq 4l+4-\dim_{\fq}(Y\cap\mu Y)
\end{equation}
and so 
\begin{equation}
\label{eq:rangehh'}
    2h-2\leq \dim_{\fq}(S\cap\mu S)\leq 2h'+2.
\end{equation}
Therefore, we get that $h\in\lbrace h', h'+1,h'+2\rbrace$. We will analyze these cases separately.\\
{\bf Case 1:} If $h'=h$, by \eqref{eq:inter} we have that $\Bar{S}\cap\mu\Bar{S}=S\cap\mu S=Y\cap\mu Y$ and so $S\cap\mu S$ is an $\F_{q^2}$-subspace. This corresponds to I) of the statement. \\
 {\bf Case 2:} If $h=h'+1$, by  \eqref{eq:rangehh'} we have that $2h'\leq\dim_{\fq}(S\cap\mu S)\leq 2h'+2$. \\
If $\dim_{\fq}(S\cap\mu S)=2h'=\dim_{\fq}(\bar{S}\cap\mu\bar{S})$, then by  \eqref{eq:inter} we get $\bar{S}\cap\mu\bar{S}=S\cap\mu S\subset Y\cap\mu Y$ and so $S\cap\mu S$ is an $\F_{q^2}$-subspace of $\fqn$. \\
If $\dim_{\fq}(S\cap\mu S)=2h'+2=\dim_{\fq}(Y\cap\mu Y)$, then $S\cap\mu S=Y\cap\mu Y$ is an $\F_{q^2}$-subspace of $\fqn$ of dimension $h'+1$. Since $\bar{S}$ is the maximal $\F_{q^2}$-subspace contained in $S$, and $\mu \bar{S}$ is the maximal $\F_{q^2}$-subspace contained  $\mu S$, then $S\cap\mu S \subseteq \bar{S}\cap \mu \bar{S}$, a contradiction to their dimensions. . This corresponds to II) of the statement.\\
 {\bf Case 3:} If $h=h'+2$, then by  \eqref{eq:rangehh'} we get  $\dim_{\fq}(S\cap\mu S)=2h'+2$. Moreover, if $S\cap\mu S$ is an $\F_{q^2}$-subspace of $\fqn$, then $S\cap\mu S\subseteq \bar{S}\cap\mu\bar{S}$,  a contradiction. Therefore $\dim_{\fq}(S\cap\mu S)=2h'+2$ and $S\cap\mu S$ is not an $\F_{q^2}$-subspace of $\fqn$. . This corresponds to III) of the statement.\\

Finally, if $\dim_{\fq}( S\cap\mu S)=2l$, then in Case III) we have that $Y\cap \mu Y=Y$ and so $\mu \in H(Y)$. In Cases I and II we get $\mu \in H(\bar{S})$.
\end{proof}

From the previous proposition, we derive a necessary condition for obtaining an FWS code starting from a subspace $S$ as in (\ref{EqBarS}). This condition will be used in Section 6 to prove Theorem 6.2.

\begin{corollary}
\label{cor:SisFWSthenbarSisFWSq2}
Let $S=\bar{S}\oplus b\fq$ be an $\fq$-subspace of $\fqn$ where $\bar{S}$ is an $\F_{q^2}$-subspace of $\fqn$ of dimension $2l$ over $\fq$. If  $\mathrm{Orb}(S) \subset \mathcal{G}_q(n,2l+1)$ is an $FWS$ code, then $\mathrm{Orb}(\bar{S})$ is an $FWS$ code in  $\mathcal{G}_{q^2}(\frac{n}2,l)$. 
Also, if $Y=\langle S\rangle_{\F_{q^2}}$, then  
 $\dim_{\F_{q^2}}(Y\cap\mu Y)$, $\mu\in\fqn^*$, admits all the possible values in $\{1,\ldots,l\}$.
\end{corollary}
\begin{proof}
Since $S$ is an FWS code, then for any $i\in \{0,1,\dots, l-1\}$ there exists $\mu_i\in\fqn\setminus\fq$ such that $\dim_{\fq}(S\cap\mu_i S)=2i+1$. Case II) of Proposition \ref{prop1 h'h} occurs. It follows that 
\[
\dim_{\fq}(\Bar{S}\cap\mu_i\Bar{S})=2i \text{ and } \dim_{\fq}(Y\cap\mu_i Y)=2i+2,
\]
so that, $\mathrm{Orb}(\bar{S})$ is an FWS code in  $\mathcal{G}_{q^2}(\frac{n}2,l)$ and 
\[ \{ \dim_{\F_{q^2}}(Y\cap\mu_i Y) \colon i \in \{0,\ldots,l-1\} \}=\lbrace 1,\dots,l\rbrace.\]
\end{proof}

In the next proposition we will determine information on the weight distribution of the cyclic orbit code associated with the subspace $S=\langle 1, \lmb, \dots, \lmb^{l-1}\rangle_{\F_{q^2}}\oplus\lmb^{l}\fq$. 
The proof relies on the use of Lemma \ref{Lem:lemmamu} and Proposition \ref{prop1 h'h}. 

\begin{proposition}
    \label{prop3}
Let $n$ a  positive even  integer,  let $\lambda \in \fqn \setminus \F_{q^2}$ and  let $Y=\langle 1, \lmb, \dots, \lmb^{l}\rangle_{\F_{q^2}}$ with  $1\leq l < \frac{t}{2}$ where $t=[\F_{q^2}(\lambda): \F_{q^2}]$.  
Let $S$ be the following $\fq$-subspace of $Y$:
$$S=\langle 1, \lmb, \dots, \lmb^{l-1}\rangle_{\F_{q^2}}\oplus\lmb^{l}\fq \in \mathcal{G}_q(n,2l+1),$$
then 
\begin{itemize}
    \item [i)] $\dim_{\fq}(S\cap\mu S)=2l$ if and only if $ \mu\in\F_{q^2}\setminus\fq$.
    \item [ii)] $\dim_{\fq}(S\cap\mu S)=2(l-r), 0<r< l$ if and only if, up to replace $\mu$ by $\mu^{-1}$, we have  $\mu=\frac{p(\lambda)}{q(\lambda)}$ where 
    $p(\lambda)=\sum_{i=0}^{r-1}\alpha_i \lambda ^i+\lmb^{r}$, $q(\lambda)=\sum_{i=0}^{r}\beta_i\lmb^i$, with $p(x),q(x) \in \F_{q^2}[x]_{\leq r}$, $\mathrm{gcd}(p(x),q(x))=1$ and  $\beta_r\in\F_{q^2}\setminus\fq$; also in this case $S\cap\mu S=\Bar{S}\cap \mu \Bar{S}$.
    \item [iii)] $\dim_{\fq}(S\cap\mu S)=2(l-r)+1$, with $0<r<l$ if and only if, up to replace $\mu$ by $\mu^{-1}$, we have $\mu=\frac{p(\lambda)}{q(\lambda)}$ where 
    $p(\lambda)=\sum_{i=0}^{r-1}\alpha_i \lambda ^i+\lmb^{r}$, $q(\lambda)=\sum_{i=0}^{r}\beta_i\lmb^i$, with $p(x),q(x) \in \F_{q^2}[x]_{\leq r}$, $\mathrm{gcd}(p(x),q(x))=1$ and    $\beta_r\in\fq$; in this case $\dim_{\fq}(S\cap\mu S)=\dim_{\fq} (\Bar{S}\cap \mu \Bar{S})+1$.
\end{itemize}
Finally, for any $r\in \{1,\dots, l-1\}$  there exist elements  $\mu \in \fqn$ for which ii) is realized and for any $r\in \{1,\dots, l\}$  there exist elements  $\mu \in \fqn$ for which iii) is realized.
\end{proposition}
\begin{proof}
Let $\bar{S}=\langle 1, \lmb, \dots, \lmb^{l-1}\rangle_{\F_{q^2}}$ and  note that by Lemma  \ref{Lem:lemmamu} $\dim_{\F_{q^2}}(\bar{S}\cap \lambda \bar{S})=l-1=\dim_{\F_{q^2}}(\bar{S})-1$ and $\dim_{\F_{q^2}}(Y\cap \lambda Y)=l=\dim_{\F_{q^2}}(Y)-1$, hence $H(Y)=H(\bar{S})=\F_{q^2}$. Suppose that $\mu\in\F_{q^2}\setminus\fq$, then
$\mu S=\langle 1, \lmb, \dots, \lmb^{l-1}\rangle_{\F_{q^2}}\oplus\mu\lmb^{l}\fq$,
hence $S\cap\mu S=\Bar{S}$ and so $\dim_{\fq}(S\cap\mu S)=2l$. On the other hand, 
if  $\dim_{\fq}(S\cap\mu S)=2l$, then by Proposition \ref{prop1 h'h} we get either  $\mu\in H(Y)=\F_{q^2}$ or  $\mu \in H(\Bar{S})=\F_{q^2}$.  This proves \emph{i)}.\\

Now we study the other possible dimensions. By Lemma \ref{Lem:lemmamu}, for any $r\in\{1,\dots,l-1\}$, we have that 
\[
\dim_{\F_{q^2}}(\Bar{S}\cap\mu \Bar{S})=l-r \text{ if and only if }\mu=\frac{p(\lmb)}{q(\lmb)}
\]
with $p(x),q(x)\in\F_{q^2}[x]_{\leq l-1} $ such that $\gcd(p(x),q(x))=1$ and $r=\max\lbrace  \deg(p(x)),\deg(q(x))\rbrace$. Up to multiply by a non-zero factor in $\F_{q^2}$ and to replace $\mu$ by $\mu^{-1}$ (as $\dim_{\fq}(\Bar{S}\cap \mu \Bar{S})=\dim_{\fq}(\mu^{-1} \Bar{S} \cap \Bar{S})$), we may suppose 
\[
\mu=\frac{p(\lmb)}{q(\lmb)}=\frac{\alpha_0+\dots+\alpha_{r-1}\lambda^{r-1}+\lmb^{r}}{\beta_0+\dots+\beta_r\lmb^r}
\]
where $\alpha_i,\beta_j\in\F_{q^2}$ for any $i\in\lbrace 0, \dots,r-1\rbrace$ and $j\in\lbrace 0, \dots,r\rbrace$.\\
Then by Lemma \ref{Lem:lemmamu}  we get that $\dim_{\F_{q^2}}(Y\cap\mu Y)=l+1-r$ and 
\begin{equation}
Y\cap\mu Y= p(\lmb)\la 1,\lmb,\ldots,\lmb^{l-r} \ra_{\F_{q^2}}= \mu q(\lmb)\la 1,\lmb,\ldots,\lmb^{l-r} \ra_{\F_{q^2}},
\end{equation}

\begin{equation}
\Bar{S}\cap\mu \Bar{S}=p(\lmb)\la 1,\lmb,\ldots,\lmb^{l-r-1} \ra_{\F_{q^2}}=\mu q(\lmb)\la 1,\lmb,\ldots,\lmb^{l-r-1} \ra_{\F_{q^2}}.
\end{equation}
Thus 
\[
2(l-r)\leq \dim_{\fq}(S\cap\mu S)\leq 2(l-r+1).
\]
In particular, note that 
\[
p(\lmb)(a_0+a_1\lmb+\dots+a_{l-r}\lmb^{l-r})=
\mu q(\lmb)(a_0+a_1\lmb+\dots+a_{l-r}\lmb^{l-r})\in (S\cap\mu S)\setminus(\Bar{S}\cap\mu\Bar{S})
\]
if and only if $a_{l-r}\in\fq^*$ and $\beta_r a_{l-r}\in\fq^*$, that is if and only if $\beta_r\in\fq^*$.
Therefore, if $\beta_r\in\fq$ then
\[
S\cap\mu S=\lbrace p(\lmb)(a_0+a_1\lmb+\dots+a_{l-r}\lmb^{l-r})\colon a_0,\ldots, a_{l-r-1}\in\F_{q^2}, a_{l-r}\in\fq^*\rbrace
\]
and so $\dim_{\fq}(S\cap\mu S)=2(l-r)+1$.\\
If $\beta_r\in\F_{q^2}\setminus\fq$, then $S\cap\mu S=\Bar{S}\cap\mu\Bar{S}$ and so $\dim_{\fq}(S\cap\mu S)=2(l-r)$.\\
Finally, let  $r\in \{1,\dots, l-1\}$.
If 
\[
\mu=\frac{\lambda^r}{\beta_0+\beta_r \lambda^r} \,\, \,\text{with} \,\, \, \beta_r \in \F_{q^2} \setminus \fq \,\, \, \text{and} \,\, \,  \beta_0 \in \F_{q^2}^*,
\]
then  by \emph{ii)}  we get $\dim_{\fq} (S\cap \mu S)=2(l-r)$; whereas  if 
\[
\mu=\frac{\lambda^r}{\beta_0+ \lambda^r} \,\, \,\text{with}  \,\, \,  \beta_0 \in \F_{q^2}^*,
\]
then  by \emph{iii)}  we get $\dim_{\fq} (S\cap \mu S)=2(l-r)+1$. 
To conclude, we need to show that there exists $\mu \in \fqn$ such that $\dim_{\fq}(S\cap \mu S)=1$. Since $l<\frac{t}{2}$, it can be easily seen that 
$$    S\cap \lambda ^l S=\lambda^l \F_{q},$$ hence $\dim_{\fq}(S\cap\lambda^l S)=1$. 
\end{proof}

As a corollary of Proposition \ref{prop3} we get the following.

\begin{corollary}
\label{cor:FWSq2}
Let $n$ be an even integer, let $\lambda \in \fqn\setminus \F_{q^2}$ and let  $t=[\F_{q^2}(\lambda): \F_{q^2}]$. Then $\cC=\mathrm{Orb}(S)$, with $S=\langle 1,\lmb,\ldots,\lmb^{l-1}\rangle_{\F_{q^2}}\oplus \lmb^l\fq$, is an $FWS$ code with $m(\cC)=2t$ in $\mathcal{G}_q(n,2l+1)$ if and only if $2l+1\leq t$.
In particular, $\omega_2(\cC)=q$.
\end{corollary}
\begin{proof}
First note that $\fq(S)=\F_{q^2}(\lambda)$ and hence $m(\cC)=2t$. If $\cC$ is an FWS code, by Corollary \ref{cor:necessaryconditionsFWS} we get that $2l+1 \leq t$. Now, suppose $2l+1 \leq t$. Then   $l < \frac{t}{2}$ and hence, by ii) and iii) of Proposition \ref{prop3}, for any $r\in \{1,\dots, l-1\}$  there exist   $\mu, \mu' \in \fqn$ such that 
$$\dim_{\fq} (S\cap \mu S)=2(l-r) \,\,\,  \text{and} \,\, \, \dim_{\fq} (S\cap \mu' S)=2(l-r)+1.$$ Also, there exists $\mu''\in \fqn$ such that $\dim_{\fq} (S\cap \mu'' S)=1$ and $\dim_{\fq} (S\cap \rho S)=2l$ if and only if $\rho \in \F_{q^2} \setminus \fq$. This implies that $H(S)=\fq$, $\omega_{2i}(\C) > 0$ for any  $i\in \{1,\dots, 2l+1\}$ and $\omega_2(\C)=\frac{q^2-q}{q-1}=q$.    This concludes the proof.
\end{proof}

Note that in the examples constructed in Corollary \ref{cor:existFWSpolbas} the value of $\omega_2$ is $q(q+1)$, whereas in the examples constructed in Corollary \ref{cor:FWSq2} it is $q$. This allows us to state the following result.

\begin{corollary}
   The examples of FWS codes constructed in Corollary \ref{cor:existFWSpolbas}  and the examples constructed in Corollary \ref{cor:FWSq2} are not linearly isometric.
\end{corollary}

In the following sections, we will prove that actually the examples of FWS codes given in Corollary \ref{cor:existFWSpolbas} and Corollary \ref{cor:FWSq2} are the only ones possible (up to the action of the group of the linear isometries).

\section{Critical pairs and a classification of cyclic subspace codes with minimum distance two}\label{sec:dist2}

In the study the weight distribution of a cyclic subspace code, it is crucial to understand the intersections of a subspace with its shifts and hence to evaluate its multiplicative behavior within the field. Such behavior is examined in certain problems that serve as linear analogues to some additive combinatorial questions. Indeed, critical pairs are classically studied in theory of combinatorial number theory and are defined as a pair of sets of a group $(G,\cdot)$ with the property that their product set has size smaller than the sum of their sizes. This is connected with Kneser's Addition Theorem and the Cauchy-Davenport inequality. Recently, linear analogues of these results have been investigated. The linear analogue of Kneser's Addition Theorem has been proved by Hou, Leung and Xiang in \cite{HouLeungXiang2002} for separable extensions fields, motivated by a problem on different sets. 

In the next we will use the following notation: let $S,T \subseteq \LL$, where $\LL$ is any field and $\F$ is a subfield of $\LL$, then
\[ ST=\{ s\cdot t \colon s \in S, t \in T\}, \]
and $\langle S T \rangle_{\F}$ denotes the $\F$-span in $\LL$ of $ST$.

\begin{theorem}\cite[Theorem 2.4]{HouLeungXiang2002} \label{teo:bachocserrazemorext}
Let $\LL/\F$ be a separable field extension and let $S$ and $T$ be  $\F$-subspaces of $\LL$ of finite positive dimension. Then

    \[
    \dim_{\F}(\langle S T \rangle_{\F})\geq \dim_{\F}(S)+\dim_{\F}(T)-\dim_{\F}(H(ST)),
    \]
    where $H(ST)=\{x\in \LL \,:\, xST\subseteq ST\}$.

\end{theorem}

In \cite{HouLeungXiang2002}, the authors started the study of the linear extensions of classical addition theorems. 
Then Bachoc, Serra and Z\'emor in \cite{BSZ2018} extended the result in \cite{HouLeungXiang2002} by removing the separability condition. 
Vosper in \cite{Vosper} proved an inverse statement of the Cauchy-Davenport inequality in additive theory by providing examples of pair of sets attaining the equality in the Cauchy-Davenport inequality, which are known as \textbf{critical pairs}.
Bachoc, Serra and Z\'emor in \cite{BSZ2015} proved a linear analogue of Vosper's result in the case of prime degree extensions of finite fields, providing a classification results for {\bf critical pairs} of $\fqn$, with $n$ prime, i.e. pairs of $\fq$-subspaces $S$ and $T$ of $\fqn$  satisfying the equality
\[
    \dim_{\fq}(\langle S T \rangle_{\fq})= \dim_{\fq}(S)+\dim_{\fq}(T)-1.
    \]

A characterization of critical pairs $(S,T)$ in extension fields $\LL/\F$ where $T$ is a two-dimensional $\F$-subspace of $\LL$ is given in the following result. 

\begin{theorem} \label{thm:extalg}\cite[Proposition 6.3]{NPSZminsize}\label{thm:classcritpairs2dim}
Let $\LL/\F$ be a field extension and let $S$ be an $\F$-subspace of $\LL$ of finite positive dimension $k$ and let $T=\langle 1,\lambda\rangle_{\F}$, for some $\lambda \in \LL\setminus \F$, such that $k+2 \leq \dim_{\F}(\LL)$ and $\lambda$ is algebraic over $\F$.
Let $\KK=\F(\lambda)$ and let $t=\dim_{\F}(\KK)$.
Suppose that $(S,T)$ is a critical pair, i.e. 
$\dim_{\F}(\langle S T \rangle_{\F})=\dim_{\F}(S)+1.$
Then one of the following cases occurs:
\begin{itemize}
    \item $t> k$ and $S=b \langle 1,\lambda,\ldots,\lambda^{k-1}\rangle_{\F}$, for some $b \in \LL^*$;
    \item $t\leq k-1$,  $S=\overline{S}\oplus b\langle 1,\lambda,\ldots,\lambda^{m-1}\rangle_{\F}$, where $\overline{S}$ is a $\KK$-subspace of dimension $\ell\geq 0$, $b \in \LL^*$, $b \KK \cap \overline{S}=\{0\}$ and   $k=t\ell+m$ with $0<m<t$.
\end{itemize}
\end{theorem}


We can use the above theorem to obtain that the one-orbit cyclic subspace codes with minimum distance two can be divided into two families, according to the different values of the invariant $\delta_t$.

\begin{theorem}\label{thm:class}
    Let $\C=\mathrm{Orb}(S) \subseteq \mathcal{G}_q(n,k)$ with $\omega_2(\C)>0$. Then there exists $\lambda \in \fqn^*$ such that $m(\C)=[\fq(\lambda)\colon \fq]$ and one of the following occurs
    \begin{itemize}
        \item[i)] $S=b \langle 1,\lmb,\ldots,\lmb^{k-1}\rangle_{\fq}$, for some $b \in \fqn^*$ if  $k<t$;
        \item[ii)]  $S=\overline{S}\oplus b\langle 1,\lmb,\ldots,\lmb^{m-1}\rangle_{\fq}$ if  $k\geq t+1$,  where $\overline{S}$ is an $\F_{q^t}$-subspace of dimension $\ell>0$, $b \in \fqn^*$, $b \F_{q^t} \cap \overline{S}=\{0\}$, $k=t\ell+m$ with $0<m<t$.
    \end{itemize}
    In particular, in Case i) we have $\delta_t(\cC)=1$ and in Case ii) we have $\delta_t(\cC)= \ell+1\geq 2$.
\end{theorem}
\begin{proof}
    Since $\omega_2(\C)>0$, then there exists $\lambda \in \fqn\setminus \fq$ such that $d( S,\lambda S)=2$, implying that 
    \[ \dim_{\fq}( S\cap \lambda S)=k-1. \]
    Observe that
    \[  S + \lambda S=  \la S\la 1,\lambda\ra_{\fq}\ra_{\fq}\] 
    and by the hypothesis, $\dim_{\fq}(S+\lambda S)=\dim_{\fq}(S)+1$.
    Therefore, $(S, \la 1,\lambda\ra_{\fq})$ is a critical pair and so, using Theorem \ref{thm:classcritpairs2dim} we get the assertion.
\end{proof}

\section{Classification of Full Weight Spectrum codes}\label{sec:proofnoFWScode}

This section is devoted to the classification of FWS one-orbit cyclic subspace codes. Let us start by observing that an FWS code must be considered within the class of one-orbit cyclic subspace codes with minimum distance two, as there cannot be any gaps in the weight distribution. In Theorem \ref{thm:class}, we divided the family of one-orbit cyclic subspace codes into two families i) and ii). Regarding Family i), we have already characterized in Corollary \ref{cor:existFWSpolbas} the FWS codes in this family. Therefore we need to analyze the FWS codes given by the subspaces of Family ii). It will be proved that a cyclic subspace code as in $ii)$ of 
Theorem \ref{thm:class} is an FWS code only when $t=2$ and $H(\langle S\rangle_{\F_{q^2}})=\F_{q^2}$, which will allow us to prove that this happens if and only if the code is defined by a subspace of the following form 
\[
S=\langle 1,\lmb,\ldots,\lmb^{l-1}\rangle_{\F_{q^2}}\oplus \lmb^l\fq,
\]
which indeed defines an FWS code as proved in Corollary \ref{cor:FWSq2}.
The proof of the classification result will proceed by some steps. Subsection \ref{subsec:proofnoFWS} is dedicated to the proof of the following theorem, which states that, the codes $\mathrm{Orb}(S)$ in Family ii), unless $\dim_{\fq}(S)=2l+1$ and $H(\langle S\rangle_{\fqt})=\F_{q^2}$, are not FWS codes. 


\begin{theorem}\label{thm:noFWScode}
    Let $\lambda \in \fqn \setminus \fq$ and $\fq(\lambda)=\fqt$. Let
    $$S=\Bar{S}\oplus b\langle 1, \lmb, \dots, \lmb^{m-1}\rangle_{\fq},$$
    where $\overline{S}$ is an $\F_{q^t}$-subspace of $\mathcal{G}_q(n,k)$ dimension $\ell>0$ of $\fqn$, $b \in \fqn \setminus \Bar{S} $ and  $0<m<t$.  If one of the following two conditions are fulfilled 
    \begin{itemize}
        \item $(t,m)\ne (2,1)$;
        \item $(t,m)=(2,1)$ and $\F_{q^2} \subset H(\langle S\rangle _{\F_{q^2}})$,
    \end{itemize}
    then $\mathrm{Orb}(S)$ is not an FWS code.
\end{theorem}

In Subsection \ref{subsec:smallcase}, we will analyze the missing case, i.e. $\dim_{\fq}(S)=2l+1$ and $H(\langle S\rangle_{\F_{q^2}})=\F_{q^2}$, proving the following theorem.

\begin{theorem}
\label{thm:casopiccoloFWS}
Let $n$ be an even positive integer and let $\Bar{S}$ an $\F_{q^2}$-subspace of $\fqn$ of dimension $l>0$. Let
$$S=\Bar{S}\oplus b\fq,$$ 
where  $\Bar{S}\cap b\F_{q^2}=\lbrace 0\rbrace$ and $H(\langle S\rangle _{\F_{q^2}}) =\F_{q^2}$. Then $Orb(S)$ is an FWS code if and only if, up to a cyclic shift, $S=\langle 1,\lmb,\ldots,\lmb^{l-1}\rangle_{\F_{q^2}}\oplus \lmb^l\fq$
    for some $\lmb\in\fqn$ such that $t:=[\F_{q^2}(\lambda) : \F_{q^2}]\geq 2l+1 $. 
\end{theorem}

The proof of the main result, cf.\ Theorem \ref{thm:maintheorem}, is detailed in Subsection \ref{subs:maintheoremfinal} and, as we will see, it is obtained by combining the results of this section and those developed in the previous sections.


\subsection{Proof of Theorem \ref{thm:noFWScode}}\label{subsec:proofnoFWS}
In the next we will prove a stronger version of Theorem \ref{thm:noFWScode}, in which we will explicitly show that some  weights in the weight distribution of the codes in Family ii) of Theorem \ref{thm:class} are zero.

\begin{theorem}
\label{thm:maintheoremjump}
Let $\mathcal{C}=\mathrm{Orb}(S)$ be the cyclic subspace code defined by $S\in\mathcal{G}_q(n,k)$ where
\[
S=\Bar{S}\oplus b\langle 1, \lmb, \dots, \lmb^{m-1}\rangle_{\fq}
\]
where $\fqt=\fq(\lmb)$ is a subfield of $\fqn$, $\Bar{S}$ is an $\fqt$-subspace of $\fqn$ of dimension $l$ such that $\Bar{S}\cap b\fqt=\lbrace 0\rbrace$ and   $0<m<t$.\\
Let $(\omega_2(\mathcal{C}),\dots,\omega_{2k}(\mathcal{C}))$ be the weight distribution of $\cC$. Also, denoted by $Y$ the $\fqt$-subspace generated by $S$, the following hold 
    \begin{itemize}
        \item[(i)] if $m<t-1$ and $H(Y)=\fqt$, then $\omega_{2m+2}(\mathcal{C})=0$;
        \item[(ii)] if $m>\frac{t+1}{2}$, then $\omega_{2(k-j)}(\mathcal{C})=0$ for any $j\in\{1,\dots,2m-t-1\}$;
        \item[(iii)] if $ \fqt \subset H(Y)$, then $\omega_{2(k-j)}(\mathcal{C})=0$ for any $j\in\{1,\dots,2m-1\}$;
        \item[(iv)] if $t=3$ and $m=2$ and $H(Y)=\F_{q^3}$, then $\omega_4(\mathcal{C})=0$;
    \end{itemize}
\end{theorem}

Note that the only missing case in Theorem \ref{thm:maintheoremjump} is the one corresponding to $(t,m)=(2,1)$ and $H(Y)=\F_{q^2}$. This is not due to a limit of the proof, indeed by Corollary \ref{cor:FWSq2}  this case actually gives other examples of FWS codes.\\

Let us start by proving {\bf (i)} of Theorem \ref{thm:maintheoremjump}.
\begin{proof}
Let $S=\Bar{S}\oplus b\langle 1, \lmb, \dots, \lmb^{m-1}\rangle_{\fq}$ and set $S_m:=\langle 1, \lmb, \dots, \lmb^{m-1}\rangle_{\fq}$. Note that $Y=\langle S\rangle_{\fqt}=\Bar{S}\oplus b \fqt$ and that $S\cap b\fqt =bS_m$.  Let consider $d(S, \mu S)$ with $\mu\in\fqn^*$.\\
Note that if $\mu\in\fqt$, then
$\Bar{S}\subseteq S\cap\mu S$,
therefore
$\dim_{\fq}(S\cap\mu S)\geq tl=k-m$
and so $d(S,\mu S)\leq 2m$.
So, suppose there exists $\mu\notin\fqt$ such that $H(Y)=\fqt$, $m<t-1$ and  
\[
d(S,\mu S)=2m+2,
\]
i.e. $\dim_{\fq}(S\cap\mu S)=lt-1$.
Note that $S\cap\mu S\subseteq Y\cap\mu Y$ and, since $Y\cap\mu Y$ is an $\fqt$-subspace of $\fqn$ of dimension $l+1$, we have that 
\[
l\leqslant \dim_{\fqt}(Y\cap\mu Y) \leqslant l+1.
\]
Moreover, since $H(Y)=\fqt$ and $\mu\notin\fqt$, $\dim_{\fq}(Y\cap\mu Y)=lt$.\\
Let observe that if $b\fqt\subseteq Y\cap\mu Y$, by Proposition \ref{prop:hyperplane}  it follows that
\[
\dim_{\fq}((S\cap \mu S)\cap b\fqt)\geq t-1,
\]
and hence $\dim_{\fq}(S\cap b\fqt)=\dim_{\fq}(S_m)=m\geq t-1$,
 a contradiction to the assumption $m<t-1$.
This implies that $b\fqt\not\subseteq Y\cap\mu Y$ and, similar arguments show that  $\mu b\fqt\not\subseteq Y\cap\mu Y$.\\
Now, since $\dim_{\fq}(S\cap\mu S)=lt-1=\dim_{\fq}(Y\cap\mu Y)-1$, Proposition \ref{prop:hyperplane} implies the existence of an $\fqt$-subspace $\Bar{T}$ of $\fqn$ such that \[
\dim_{\fqt}(\Bar{T})=l-1 \quad \quad \mbox{and} \quad \quad \Bar{T}\subseteq S\cap\mu S. \]

Note that $\Bar{S}$ is the maximal $\fqt$-subspace of $\fqn$ contained in $S$ and $\mu\Bar{S}$ is the maximal $\fqt$-subspace of $\fqn$ contained in $\mu S$, then $\Bar{T}\subseteq\Bar{S}$ and   $\Bar{T}\subseteq\mu \Bar{S}$. Then 
\[
\Bar{T}\subseteq \Bar{S}\cap\mu\Bar{S}\subseteq \Bar{S}
\]
and so 
\[
l-1=\dim_{\fqt}(\Bar{T})\leqslant \dim_{\fqt}(\Bar{S}\cap\mu\Bar{S})\leqslant \dim_{\fqt}(\Bar{S})=l.
\]
If $\dim_{\fqt}(\Bar{S}\cap\mu\Bar{S})=l$, then $\Bar{S}=\mu\Bar{S}$ and this implies that $\Bar{S}\subseteq S\cap\mu S$. Then $lt=\dim_{\fq}(\Bar{S})\leqslant \dim_{\fq}(S\cap\mu S)=lt-1$, a contradiction. Therefore $\dim_{\fqt}(\Bar{S}\cap\mu\Bar{S})=l-1=\dim_{\fqt}(\Bar{T})$ and hence
\[
\Bar{T}=\Bar{S}\cap\mu\Bar{S}.
\]
Since $\dim_{\fqt}(\Bar{T})=l-1=\dim_{\fqt}(Y\cap\mu Y)-1$, there exists $\xi\fqt\subseteq Y\cap\mu Y$ such that $\xi\fqt\not\subseteq \Bar{T}$. By Proposition \ref{prop:hyperplane} it follows that
\[
t-1\leq \dim_{\fq}(\xi\fqt\cap S\cap\mu S)\leq t.
\]
Moreover, $\dim_{\fq}(S\cap\mu S\cap\xi\fqt)\neq t$ by the maximality of $\Bar{S}$ and $\mu \bar{S}$ as $\fqt$-subspaces of   $S$ and $\mu S$ respectively. 
  Therefore  $\dim_{\fq}(S\cap\mu S\cap\xi\fqt)=t-1$.
Therefore, 
\[
\xi H=\xi\fqt\cap S\cap\mu S,
\]
where $H$ is an $\fq$-subspace of $\fqt$ with dimension $t-1$, and so
\[
S\cap\mu S=\Bar{T}\oplus\xi H.
\]
Moreover, since $\xi\in Y=\Bar{S}\oplus b\fqt$, then there exist $\Bar{s}\in\Bar{S}$ and $\alpha\in\fqt$ such that 
$\xi=\Bar{s}+b \alpha$.
Now, if we take any $h\in H$ and $\Bar{t}\in\Bar{T}$, we have $\eta= \bar{t}+\xi h \in S\cap\mu S=\Bar{T}\oplus\xi H$, and 
\[
\eta=\bar{t}+\xi h=\Bar{t}+(\Bar{s}+b \alpha)h=(\bar{t}+\Bar{s}h)+b \alpha h.
\]
Since $\eta\in S$, there exist $\Bar{s_1}\in\Bar{S}$ and $s_m\in S_m$ such that 
\[
\eta=\Bar{s}_1+b s_m.
\]
Then we have
\[
(\bar{t}+\Bar{s}h)+b\alpha h=\Bar{s}_1+b s_m
\]
and so
\[
(\bar{t}+\Bar{s}h)-\Bar{s}_1=b (s_m-\alpha h).
\]
The left hand side belongs to $\Bar{S}$, the right hand side belongs to $b\fqt$. Since $\Bar{S}\cap b\fqt=\lbrace 0\rbrace$, we immediately get
\[
\alpha h=s_m \text{ for any } h\in H.
\]
Then we have
\[
\alpha H\subseteq S_m.
\]
If $\alpha\neq 0$ this implies that $t-1=\dim_{\fq}(\alpha H)\leq \dim_{\fq}(S_m)=m<t-1$, a contradiction. Therefore $\alpha=0$ and hence $\xi\in\Bar{S}$. 
Similarly, since $\xi\in \mu Y=\mu\Bar{S}\oplus\mu b\fqt$, then there exist $\Bar{s'}\in\Bar{S}$ and $\beta\in\fqt$ such that 
\[
\xi=\mu\Bar{s}'+\mu b\beta
\]
and by arguing as before, we get $\beta=0$ and so $\xi\in\mu\Bar{S}$. Then we have
$\xi\fqt\subseteq\Bar{S}\cap\mu\bar{S}=\Bar{T}$,
a contradiction since $\xi\fqt\not\subseteq\Bar{T}$. 
\end{proof}

Now, let us prove {\bf (ii) and (iii)} of Theorem \ref{thm:maintheoremjump}.
\begin{proof}
Note that $\dim_{\fqt}(Y)=l+1$ and hence $\dim_{\fq}(Y)=t(l+1)$. Moreover, since $S\subseteq Y$, we have 
\[
\dim_{\fq}(Y+\mu Y)\geq \dim_{\fq}(S+\mu S)= 2k -\dim_{\fq}(S \cap \mu S)
\]
and then
\begin{equation}
\label{eq:case2.12.2}
 \dim_{\fq}(S\cap\mu S)\geq 2m-2t+\dim_{\fq}(Y\cap\mu Y).   
\end{equation}
Let $\dim_{\fq}(S\cap \mu S)>0$. Then,  since   
$Y\cap\mu Y$ is an $\fqt$-subspace of $\fqn$ and $S\cap\mu S \subseteq  Y\cap\mu Y$, we get  
$$\dim_{\fq}(Y\cap\mu Y)=ht\geq t.$$
In {\bf Case (ii)}, $m>\frac{t+1}{2}$, hence $2m-t>1$ and so by  \eqref{eq:case2.12.2} it follows that
\[
\dim_{\fq}(S\cap\mu S)\geq ht-2t+2m\geq 2m-t,
\] 
and so $\omega_{2(k-j)}(\cC)=0$ for any $j\in\{1,\dots,2m-t-1\}$.\\
In {\bf Case (iii)}, $\fqt \subset H(Y)$, hence  $H(Y)=\F_{q^{ts}}$ for some $s>1$ such that $t\cdot  s\mid n$ and so $Y$ and $Y\cap\mu Y$  are $H(Y)$-linear subspaces of $\fqn$  for any $\mu\in\fqn$. For those $\mu \in \fqn$ such that $\dim_{\fq}(S\cap \mu S)>0$, we get that $\dim_{\fq}(Y\cap\mu Y)\geq t    \cdot s\geq 2t$. Then by using again \eqref{eq:case2.12.2} it follows that
\[
\dim_{\fq}(S\cap \mu S)\geq 2m,
\]
implying that $\omega_{2(k-j)}(\cC)=0$ for any $j\in\{1,\dots,2m-1\}$.
\end{proof}

In order to prove {\bf (iv)} of Theorem \ref{thm:maintheoremjump}, we need to prove the following lemma.

\begin{lemma}
    \label{lem:case2.4}
Let $S\in\mathcal{G}_q(n,k)$ such that
\[
S=\Bar{S}\oplus b\langle 1, \lmb, \dots, \lmb^{m-1}\rangle_{\fq}\subseteq \Bar{S}\oplus b \fqt =Y
\]
where $\Bar{S}$ is an $\fqt$-subspace of $\fqn$ of dimension $l$, with $\fqt=\fq(\lmb)$, such that $\Bar{S}\cap b\fqt=\lbrace 0\rbrace$, $k=tl+m$ and $0<m<t$. Then
\begin{itemize}
    \item If $\mu \in \F_{q^t}^*$ then $\dim_{\fq}(S\cap\mu S)\geq lt$.
    \item If $\dim_{\fq}(S\cap\mu S)\geq lt$, then $\mu\in H(Y)$.
\end{itemize}
\end{lemma}
\begin{proof}
 For the first item, just note that in this case $\Bar{S}$ is an $\fqt$-subspace contained in $S\cap \mu S$.
 For the second part, suppose by contradiction that $\dim_{\fq}(S\cap \mu S)\geq lt$ and $\mu\notin H(Y)$. \\
 Then,  since $S\cap\mu S\subseteq Y\cap\mu Y\subseteq Y$, we get
 \[
 l\leq \dim_{\fqt}(Y\cap\mu Y)\leq l+1.
 \]
 Since $\mu\notin H(Y)$, we have that 
 \[
 \dim_{\fqt}(Y\cap\mu Y)=l
 \]
 and so
 \[
 S\cap \mu S=Y\cap \mu Y \text{ and } \dim_{\fq}(S\cap \mu S)=lt=\dim_{\fq}(\Bar{S}).
 \]
 Note that $\Bar{S}$ is the maximal $\fqt$-subspace contained in $S$, therefore 
 \[
 S\cap\mu S=Y\cap \mu Y= \Bar{S}.
 \]
 In particular, $\dim_{\fqt}(Y\cap\mu Y)=l=\dim_{\fqt}(Y)-1$, i.e. $Y\in \mathcal{G}_{q^t}(n/t,l)$ and $d(Y, \mu Y)=2$, then by Theorem \ref{thm:class} we have that
 \[
 Y=\Bar{T}\oplus\xi\langle 1,\mu,\dots,\mu^{c-1}\rangle_{\fqt}
 \]
 where $\Bar{T}$ is a subspace over $\fqt(\mu)=\F_{q^{ts}}$ with $\dim_{\F_{q^{ts}}}(\Bar{T})=h\geq 0$, where $t \cdot s\mid n, t\cdot s<n$, $c<s$ and $\Bar{T} \cap \xi \F_{q^{ts}}=\{0\}$. In particular, 
 \[
 \dim_{\F_{q}}(Y)=(l+1)t=(sh+c)t.
 \]
 Thus we have that
 \[
 \mu Y= \Bar{T}\oplus\xi\langle \mu, \dots, \mu^c\rangle_{q^t}
 \]
 and so 
 \[
\Bar{T}+\xi\langle \mu, \dots, \mu^{c-1}\rangle_{\fqt}\subseteq Y\cap \mu Y.
 \]
Since $\dim_{\fqt}(Y\cap\mu Y)=l=sh+c-1=\dim_{\fqt}(\Bar{T}+\xi\langle \mu, \dots, \mu^{c-1}\rangle_{\fqt})$, we get that
  \begin{equation}
  \label{eq:ymuYsmuSbarS}
 Y\cap \mu Y=\Bar{T}+\xi\langle \mu, \dots, \mu^{c-1}\rangle_{\fqt}=S\cap\mu S=\Bar{S}.
 \end{equation}
 Since $S=\Bar{S}\oplus b\langle 1, \lmb,\dots,\lmb^{m-1}\rangle_{\fq}=\Bar{S}\oplus b S_m$, we have that
 \begin{equation}
 \label{eq:S1}
 S=\Bar{T}\oplus \xi\langle \mu, \dots, \mu^{c-1}\rangle_{\fqt}\oplus b S_m.
 \end{equation}
 Note that $b S_m=b \langle 1, \lmb, \dots, \lmb^{m-1}\rangle_{\fq}\subseteq Y=\Bar{T}\oplus\xi\langle 1,\mu,\dots,\mu^{c-1}\rangle_{\fqt}$, therefore for any $i\in\{0,\dots,m-1\}$ there exist $\Bar{t}_i\in\Bar{T}$ and $a_{i,j}\in\fqt$ for $j\in\{0,\dots,c-1\}$ such that 
 \[
 b \lmb^i=\Bar{t}_i+\xi(a_{i,0}+a_{i,1}\mu+\dots+a_{i,c-1}\mu^{c-1}).
 \]
 Note that, since $b S_m\cap \Bar{S}=\lbrace 0\rbrace$, $a_{i,0}\neq 0$ for any $i\in \{0,\dots,m-1\}$. Moreover, $\dim_{\fq}(S)=tl+m=t(sh+c-1)+m$, thus we have
 \begin{equation}
 \label{eq:S}
 S= \Bar{T}\oplus\xi\langle \mu, \dots, \mu^{c-1}\rangle_{\fqt}\oplus\xi \langle a_{0,0},\dots,a_{0,m-1}\rangle_{\fq}
 \end{equation}
 and so
 \begin{equation}
 \label{eq:muS}
 \mu S= \Bar{T}\oplus\xi\langle \mu^2, \dots, \mu^{c}\rangle_{\fqt}\oplus\xi\mu \langle a_{0,0},\dots,a_{0,m-1}\rangle_{\fq}.
 \end{equation}
 Note that, since the $a_{i,0}$'s are in $\fqt$, we have that 
 \[
 \langle a_{0,0},\dots,a_{0,m-1}\rangle_{\fq}\subseteq\fqt,
 \] therefore 
 \[
 \xi\mu\langle a_{0,0},\dots,a_{0,m-1}\rangle_{\fq}\subseteq \xi\langle \mu, \dots, \mu^{c-1}\rangle_{\fqt}\subseteq S.
\]
This implies that 
 \[
 \Bar{T}\oplus\xi\langle \mu^2,\dots,\mu^{c-1}\rangle_{\fqt}\oplus\xi\mu\langle a_{0,0},\dots,a_{0,m-1}\rangle_{\fq}\subseteq S\cap \mu S.
 \]
 On the other hand, suppose that $w\in S\cap\mu S$, then by Equations \eqref{eq:S} and \eqref{eq:muS} there exist $\alpha_1,\alpha_i,\alpha_i', \alpha_c\in\fqt$ for $i\in\{2,\dots,c-1\}$, $\beta_j,\beta_j'\in\fq$ for $j\in \{1,\dots,m-1\}$ and $t,t'\in\Bar{T}$ such that
 \begin{equation}
     \begin{aligned}
         w&= t+ \xi (\alpha_1\mu+\dots+\alpha_{c-1}\mu^{c-1})+\xi(\beta_0a_{0,0}+\dots+\beta_{m-1}a_{0,m-1})\\
         &=t'+\xi(\alpha_2'\mu^2+\dots+\alpha_c'\mu^c)+\xi\mu(\beta_0'a_{0,0}+\dots+\beta_{m-1}'a_{0,m-1})
     \end{aligned}
 \end{equation}
 Then, since $\Bar{T}\cap \xi \F_{q^{ts}}=\lbrace 0\rbrace$, it follows that
 \begin{equation}
 \begin{cases}
     t=t'\\
     \beta_0a_{0,0}+\dots+\beta_{m-1}a_{0,m-1}=0\\
     \alpha_1=\beta_0'a_{0,0}+\dots+\beta_{m-1}'a_{0,m-1}\\
     \alpha_2=\alpha_2'\\
     \vdots\\
     \alpha_{c-1}=\alpha_{c-1}'\\
     \alpha_c'=0
     
 \end{cases}
 \end{equation}
 and so
 \[
w=t+\xi(\alpha_2\mu^2+\dots+\alpha_{c-1}\mu^{c-1})+\xi\mu(\beta_0'a_{0,0}+\dots+\beta_{m-1}'a_{0,m-1}).
 \]
Thus it follows that
\begin{equation}
\label{eq:S2}
 S\cap \mu S=\Bar{T}\oplus\xi\langle \mu^2,\dots,\mu^{c-1}\rangle_{\fqt}\oplus\xi\mu\langle a_{0,0},\dots,a_{0,m-1}\rangle_{\fq}.
 \end{equation}
 By Equation \eqref{eq:S2} we get that
 \[
\dim_{\fq}(S\cap\mu S)= lt=(sh+c-1)t=hts+(c-2)t+m
 \]
and so $m=t$, a contradiction. Therefore, $\mu\in H(Y)$.
\end{proof}

Now, {\bf Case (iv)} follows as an immediate consequence of Lemma \ref{lem:case2.4}.

\begin{proof}
Let $t=3$, $m=2$ and $H(Y)=\F_{q^3}$. Let $\mu\in\fqn \setminus \fq$ such that $\dim_{\fq}(S\cap\mu S)\geq 3l=k-2$. Then by Lemma \ref{lem:case2.4} we have that $\mu\in H(Y)\setminus\fq=\F_{q^3}\setminus\fq$ and so since $\bar S$ is $\F_{q^3}$-linear and  $\mu \langle 1,\lambda \rangle \subset \F_{q^3}$, we have that 
\[
\dim_{\fq}(S\cap\mu S)=k-1=3l+1
\]
and so $\omega_4(\C)=0$.
\end{proof}

Theorem \ref{thm:maintheoremjump} is now proved and, as a consequence, we get also Theorem \ref{thm:noFWScode}. 
Indeed, if $\mathbb{F}_{q^t} \subset H(Y)$ then (iii) of Theorem \ref{thm:maintheoremjump} implies that $\C$ is not an FWS code.
Assume that $H(Y)=\mathbb{F}_{q^t}$, then (i) of Theorem \ref{thm:maintheoremjump} implies that if $m<t-1$ then $\C$ is not an FWS code. Therefore, it remains to analyze the case in which $m=t-1$. By (ii) of Theorem \ref{thm:maintheoremjump}, if $t>3$ then $\C$ is not an FWS code. (iv) of Theorem \ref{thm:maintheoremjump} implies that when $(t,m)=(3,2)$ then $\C$ is not an FWS code.
In the next subsection we will analyze the missing case, i.e. $(t,m)=(2,1)$ and $H(Y)=\F_{q^2}$.

\subsection{Proof of Theorem \ref{thm:casopiccoloFWS}}\label{subsec:smallcase}

In this subsection we prove Theorem \ref{thm:casopiccoloFWS}, in order to do so we need to deal with some intermediate steps and auxiliary results.

\begin{proposition}
    \label{prop:pushup1}
Let $\Bar{S}=a \langle 1, \lmb, \dots, \lmb^{l-1}\rangle_{\F_{q^2}}$ be an $\F_{q^2}$-subspace of $\fqn$ with $2l+1\leq [\F_{q^2}(\lmb)\colon \fq]$. Let $b\in\fqn^*$ and $S=\Bar{S} + b\fq$ such that $b\F_{q^2}\cap\Bar{S}=\lbrace 0 \rbrace$. If there exists $\mu\in\fqn\setminus\fq$ such that $\dim_{\fq}(S\cap\mu S)=2l-1$, then, up to multiply by a non-zero scalar in $\F_{q^2}$, $S$ is one of the following:
\begin{itemize}
    \item [1)] $S=a(\langle 1, \lmb, \dots, \lmb^{l-1}\rangle_{\F_{q^2}}\oplus\lmb^{l}\fq)$;
    \item [2)] $S=a \rho^{l-1} (\langle 1, \Bar{\lambda}, \dots, \Bar{\lambda}^{l-1}\rangle_{\F_{q^2}}\oplus\Bar{\lambda}^{l}\fq)$, where $\rho=f+g\lambda$, $\Bar{\lambda}=\frac{1}{\rho}$,  and $f,g\in \F_{q^2}$ with $g\neq 0$.
\end{itemize}
\end{proposition}
\begin{proof}
By hypothesis, $\dim_{\fq}(S\cap\mu S)=2l-1$, then by Proposition \ref{prop1 h'h} it follows that $\dim_{\fq}(\Bar{S}\cap\mu\Bar{S})=2l-2$ and $\dim_{\fq}(Y\cap\mu Y)=2l$, where $Y=\Bar{S}\oplus b \F_{q^2}$. Moreover, since $\Bar{S}=a \langle 1, \lmb, \dots, \lmb^{l-1}\rangle_{\F_{q^2}}$ and $\dim_{\F_{q^2}}(\Bar{S}\cap\mu \Bar{S})=l-1=\dim_{\F_{q^2}}(\Bar{S})-1$, by Lemma \ref{Lem:lemmamu} we have that $\mu=\frac{\alpha+\beta \lmb}{\gamma+\delta \lmb}$ with $\alpha,\beta,\gamma,\delta\in\F_{q^2}$ such that $\alpha \delta-\beta \gamma \neq 0$. Also,
\begin{equation}
    \begin{aligned}
      \Bar{S}\cap\mu \Bar{S}&=a\lbrace (\alpha+\beta \lmb)(a_0+\dots+a_{l-2}\lmb^{l-2}) \colon a_i\in\F_{q^2}, i\in\lbrace 0, \dots, l-2\rbrace\rbrace\\
      &=a\mu\lbrace (\gamma+\delta \lmb)(a_0+\dots+a_{l-2}\lmb^{l-2}) \colon a_i\in\F_{q^2}, i\in\lbrace 0, \dots, l-2\rbrace\rbrace \subseteq S\cap\mu S.
    \end{aligned}
\end{equation}
Since $\dim_{\fq}(S\cap\mu S)=2l-1$, we have
\[
S\cap\mu S=(\bar{S}\cap\mu\Bar{S})\oplus\xi\fq
\]
for some $\xi\in S\cap\mu S$. This implies that there exist $r(x),r'(x)\in\F_{q^2}[x]_{\leq l-1}, w,w'\in\fq^*$ such that 
\[
\xi=ar(\lmb)+bw=\mu(ar'(\lmb)+bw')
\]
and so
\[
a(r(\lmb)-\mu r'(\lmb))=b(\mu w'-w).
\]
Note that $\mu w'-w\neq 0$ otherwise $\mu$ would be in $\fq$. So,
\[
b=a\frac{r(\lmb)-\mu r'(\lmb)}{\mu w'-w}.
\]
Then we have 
\[
b=a\frac{r(\lmb)(\gamma+\delta \lmb)-(\alpha+\beta \lmb)r'(\lmb)}{(\alpha+\beta\lmb) w'-(\gamma+\delta\lmb)w}
\]
and so
\[
\frac{b}{a}=\frac{q(\lmb)}{f+g\lmb}
\]
where $q(x)\in\F_{q^2}[x]_{\leq l}$, $f=\alpha w'-\gamma w, g=\beta w'-\delta w$ and $(f,g)\neq(0,0)$ because $\mu\notin\fq$. This implies that 
\[
S=a\langle 1, \lmb, \dots, \lmb^{l-1}\rangle_{\F_{q^2}}\oplus b\fq=a\left(\langle 1, \lmb, \dots, \lmb^{l-1}\rangle_{\F_{q^2}}\oplus\frac{q(\lmb)}{f+g\lmb}\fq\right).
\]
If $g=0$, since $f\in\F_{q^2}^*$, then
\[
a^{-1}f S=\langle 1, \lmb, \dots, \lmb^{l-1}\rangle_{\F_{q^2}}\oplus q(\lmb)\fq.
\]
Also, note that $q(\lmb)=a_0+a_1\lmb+\dots+a_l\lmb^l$ for some $a_0,\ldots,a_l\in\F_{q^2}$. In particular, $a_l\neq 0$ because otherwise $\dim_{\fq}(S)<2l+1$. Therefore,
\[
a_l^{-1}a^{-1}f S=\langle 1, \lmb, \dots, \lmb^{l-1}\rangle_{\F_{q^2}}\oplus \lmb^l\fq,
\]
and hence, since $a_l, f\in \F_{q^2}$, we get Case 1).\\

If $g\neq 0$, denote by $\rho=f+g\lmb$, then $\lmb=\frac{\rho-f}{g}$ and so
\[
a^{-1}(f+g\lmb)S=(f+g\lmb)\langle 1, \lmb, \dots, \lmb^{l-1}\rangle_{\F_{q^2}}\oplus q(\lmb)\fq.
\]
Moreover, $q(\lmb)=q(\frac{\rho-f}{g})=\bar{q}(\rho)$ and $\deg(\bar{q}(x))=\deg(q(x))\leq l$. Therefore
\[
a^{-1}\rho S=\langle \rho, \dots, \rho^l\rangle_{\F_{q^2}}\oplus \bar{q}(\rho)\fq
\]
and $\bar{q}(\rho)=b_0+b_1\rho+\dots+b_l\rho^l$ for some $b_0,\ldots,b_l\in\F_{q^2}$ and $b_0\neq 0$ otherwise $\dim_{\fq}(S)<2l+1$. Thus
\[
a^{-1}b_0^{-1}\rho S=\langle \rho, \dots, \rho^l\rangle_{\F_{q^2}}\oplus \fq.
\]

Finally, it is easy to check that
\[
\frac{1}{\rho^{l}} \left(\langle \rho,\rho^2,\ldots,\rho^{l}\rangle_{\F_{q^2}}\oplus \fq\right)=\langle 1,\bar{\lambda},\ldots,\bar{\lambda}^{l-1}\rangle_{\F_{q^2}}\oplus \bar{\lambda}^l\fq
\]
where $\bar{\lambda}=\frac{1}{\rho}$ and hence Case 2) follows.
\end{proof}



The next proposition shows that an FWS one-orbit cyclic subspace code in $\mathcal{G}_{q^2}(\frac{n}{2},k)$, from Family 2) of the main theorem, cannot be further extended  while preserving the property of being FWS.

\begin{proposition}
    \label{prop:pushup2}
Let $\F_{q^4}(\lmb)$ be a subfield of $\fqn$ and let 
\[S= \langle 1, \lmb, \dots, \lmb^{h-1}\rangle_{\F_{q^4}}\oplus\lmb^h\F_{q^2} \oplus b\fq,\]
with $ b\in\fqn^*$ and $\dim_{\fq}(S)=4h+3$. Let $\Bar{S}=\langle 1, \lmb, \dots, \lmb^{h-1}\rangle_{\F_{q^4}}\oplus\lmb^h\F_{q^2}$ and suppose that 
$\mathrm{Orb}(\Bar{S}) \in \mathcal{G}_{q^2}(\frac{n}{2},2h+1)$ is an FWS code.
Then $\omega_4(\mathrm{Orb}(S))=0$, and so $\mathrm{Orb}(S)$ is not an FWS code in $\mathcal{G}_{q}(n,4h+3)$.
\end{proposition}
\begin{proof}
Suppose by contradiction that $\omega_4(\mathrm{Orb}(S))\ne 0$ and so there exists $\mu\in\fqn$ such that $\dim_{\fq}(S\cap\mu S)=4h+1$. Let $Y=\langle S \rangle_{\F_{q^2}}=\langle 1, \lmb, \dots, \lmb^{h-1}\rangle_{\F_{q^4}}\oplus\lmb^h\F_{q^2}\oplus  b\F_{q^2}=\Bar{S}\oplus b\F_{q^2}$. Then by Proposition \ref{prop1 h'h}, it follows that 
\[
\dim_{\fq}(\Bar{S}\cap\mu \Bar{S})=4h \text{ and } \dim_{\fq}(Y\cap\mu Y)=4h+2.
\]
Note that $\Bar{S}=\Bar{\Bar{S}}\oplus\lmb^h\F_{q^2}$ where $\Bar{\Bar{S}}=\langle 1, \lmb, \dots, \lmb^{h-1}\rangle_{\F_{q^4}}$ and  $\dim_{\F_{q^2}}(\bar{S}\cap\mu \bar{S})=2h=\dim_{\F_{q^2}}(\bar{S})-1$. Then by applying again Proposition \ref{prop1 h'h} we have that 
\[ 
\, \mbox{either} \,\, \, \mu\in H(\Bar{\Bar{S}})=\F_{q^4} \,\, \, \mbox{or} \,\, \mu \in H(\langle\Bar{S}\rangle_{\F_{q^4}}).
\]
Since $\mathrm{Orb}(\Bar{S})$ is an FWS code, by (iii) of Theorem \ref{thm:maintheoremjump}, it follows that $H(\langle\Bar{S}\rangle_{\F_{q^4}})=\F_{q^4}$.
Therefore, we get $\mu \in \F_{q^4} \setminus\F_{q^2}$ and $\Bar{S}\cap\mu \Bar{S}= \Bar{\Bar{S}}$. Moreover, there exists $\xi\in(S\cap\mu S)\setminus(\Bar{S}\cap\mu\Bar{S})$, i.e. there exist $p(x),q(x)\in\F_{q^4}[x]_{\leq h-1}, w',w\in\F_{q^2}^*$ and $\beta,\beta'\in\fq^*$ such that
\[
\xi=p(\lmb)+\lmb^h w+b\beta=\mu(q(\lmb)+\lmb^hw'+b\beta').
\]
This implies that
\[
(p(\lmb)-\mu q(\lmb))=\lmb^h(\mu w'-w)+b(\mu \beta'-\beta).
\]
Note that $\mu\beta'-\beta\neq 0$ and $\mu w'-w\neq 0$, since $\mu\in\F_{q^4}\setminus\F_{q^2}$. So that
\[
b=\frac{p(\lmb)-\mu q(\lmb)}{\mu\beta'-\beta}-\lmb^h\frac{\mu w'-w}{\mu \beta'-\beta}\in \langle 1,\lmb, \lmb^2,....\lmb^h \rangle_{\F_{q^4}}
\]
because $\frac{p(\lmb)-\mu q(\lmb)}{\mu\beta'-\beta}\in \Bar{\Bar{S}}$ and $\rho=\frac{\mu w'-w}{\mu \beta'-\beta}\in\F_{q^4}^*$. This implies that
\[
S=\langle 1, \lmb, \dots, \lmb^{h-1}\rangle_{\F_{q^4}}\oplus \lmb^h\F_{q^2}\oplus \lmb^h\rho\fq
\]
and $\rho\notin\F_{q^2}$ because otherwise $\dim_{\fq}(S)<4h+3$.  
Denote by $U=\F_{q^2}\oplus\rho\fq $, we have
\[
S=\Bar{\Bar{S}}\oplus \lmb^hU \text{ and } \mu S=\Bar{\Bar{S}}\oplus\mu\lmb^h U.
\]
Since $\dim_{\fq} (U)=3$ and $U\subset \F_{q^4}$, we get $\dim_{\fq}(U\cap \mu U)\geq 2$ and hence $\dim_{\fq}(S\cap \mu S)\geq 4h+2$, a contradiction. 
This means that  $\dim_{\fq}(S\cap \mu S)\neq 4h+1$ for any $\mu \in \fqn$ and hence $d(S,\mu S)\neq 4$ for any $\mu \in \fqn$.
\end{proof}

Now we are ready to prove Theorem \ref{thm:casopiccoloFWS}.
\begin{proof}
Note that, since $\mathrm{Orb}(S) \subset \mathcal {G}_{q}(n,2l+1)$ is an FWS code, then by Corollary \ref{cor:SisFWSthenbarSisFWSq2}  $\mathrm{Orb}(\Bar{S})$ is an FWS code in $\mathcal {G}_{q^2}(\frac{n}{2},l)$. 
Hence by Theorem \ref{thm:class},  Theorem \ref{thm:noFWScode} and  Corollary \ref{cor:SisFWSthenbarSisFWSq2}, it follows that one of the following occurs
\begin{itemize}
    \item [$1.1)$] $\bar{S}=a\langle 1, \lmb, \dots, \lmb^{l-1}\rangle_{\F_{q^2}}$ and  $l\leq \frac{[\F_{q^2}(\lambda)\colon\F_{q^2}]+1}{2}$;
    \item[$1.2)$] $\bar{S}=\bar{S}_1\oplus b_1\F_{q^2}$ where $\Bar{S}_1$ is an $\F_{q^4}$-subspace of dimension $l_1=\frac{l-1}{2}$ over $\F_{q^4}$,  $b_1\F_{q^4}\cap \bar{S}_1=\lbrace 0\rbrace$, $H(\langle \bar{S} \rangle_{\F_{q^4}})=\F_{q^4}$ and $\mathrm{Orb}(\bar{S}_1)$ is an FWS code in $\mathcal {G}_{q^4}(\frac{n}{4},\frac{l-1}{2})$.
\end{itemize}

In Case 1.1), we get
$$S=a\langle 1, \lmb, \dots, \lmb^{l-1}\rangle_{\F_{q^2}}+b\fq.$$
Now, applying Proposition \ref{prop:pushup1}  and Corollary \ref{cor:FWSq2} we obtain the assertion.\\
In Case 1.2), since $\mathrm{Orb}(\bar{S}_1)$ is an FWS code in $\mathcal {G}_{q^4}(\frac{n}{4},\frac{l-1}{2})$, there exists $\lmb_1 \in \fqn$ such that $d(\Bar{S}_1,\lmb_1 \Bar{S}_1)=2$.  Hence by Theorems \ref{thm:class} and \ref{thm:noFWScode}, it follows that one of the following occurs
\begin{itemize}
    \item [$2.1)$] $\bar{S_1}=a_1\langle 1, \lmb_1, \dots, \lmb_1^{l-1}\rangle_{\F_{q^4}}$ and $\dim_{\F_{q^4}}(\bar{S}_1)\leq \frac{[\F_{q^4}(\lmb_1)\colon \F_{q^2}]+1}{2}$;
    \item[$2.2)$] $\bar{S}_1=\bar{S}_2\oplus b_2\F_{q^4}$ where $\Bar{S}_2$ is an $\F_{q^8}$-subspace of dimension $l_2=\frac{l_1-1}{2}=\frac{l-3}{4}$ over $\F_{q^8}$,  $b_2\F_{q^8}\cap \bar{S}_2=\lbrace 0\rbrace$, $H(\langle \bar{S}_1 \rangle_{\F_{q^8}})=\F_{q^8}$ and $\mathrm{Orb}(\bar{S}_2)$ is an FWS code in $\mathcal {G}_{q^8}(\frac{n}{8},\frac{l-3}{4})$.
\end{itemize}
In Case $2.1)$, we get
$$\Bar{S}=a_1\langle 1, \lmb_1, \dots, \lmb_1^{l-1}\rangle_{\F_{q^4}}\oplus b_1\F_{q^2}.$$
Now, applying Propositions \ref{prop:pushup1} to the $\F_{q^4}$-subspace $a_1\langle 1, \lmb_1, \dots, \lmb_1^{l-1}\rangle_{\F_{q^4}}$, we get that there exist $\Bar{a}_1$ and $\Bar{\lambda}_1$ in $\fqn$ such that
$$\Bar{S}=\Bar{a}_1 (\langle 1, \Bar{\lambda}_1, \dots, \Bar{\lambda}_1^{l-1}\rangle_{\F_{q^4}} \oplus \Bar{\lambda}_1^{l} \F_{q^2}),$$
and hence 
$$\Bar{a}_1^{-1}S=\langle 1, \Bar{\lambda}_1, \dots, \Bar{\lambda}_1^{l-1}\rangle_{\F_{q^4}} \oplus \Bar{\lambda}_1^{l} \F_{q^2}\oplus b \Bar{a}_1^{-1}\fq,$$
 and by Proposition \ref{prop:pushup2}, we get that $\mathrm{Orb}(S)$ is not an FWS code, a contradiction. Hence Case $2.1)$ does not occur.\\
 For the Case $2.2)$
 $$ S=\bar{S}_2\oplus b_2\F_{q^4} \oplus b_1 \F_{q^2} \oplus b\fq,$$  and $\mathrm{Orb}(\bar{S}_2)$ is an FWS code in $\mathcal{G}_{q^8}(\frac{n}{8},\frac{l-3}{4})$ and $\mathrm{Orb}(\bar{S}_1)$ is an FWS code in $\mathcal{G}_{q^4}(\frac{n}{4},\frac{l-1}{2})$.
 Hence we may apply again Theorems \ref{thm:class} and \ref{thm:noFWScode} to the $\F_{q^8}$-subspace $\bar{S}_2$, but now by the previous arguments we only get the following case
 
 \begin{itemize}
    \item[$3)$] $\bar{S}_2=\bar{S}_3\oplus b_3\F_{q^{8}}$ where $\Bar{S}_3$ is an $\F_{q^{16}}$-subspace of dimension $l_3=\frac{l_2-1}{2}=\frac{l-7}{8}$ over $\F_{q^{16}}$,  $b_3\F_{q^{16}}\cap \bar{S}_3=\lbrace 0\rbrace$, $H(\langle \bar{S}_2 \rangle_{\F_{q^{16}}})=\F_{q^{16}}$ and $\mathrm{Orb}(\bar{S}_3)$ is an FWS code in $\mathcal {G}_{q^{16}}(\frac{n}{16},\frac{l-7}{8})$.
\end{itemize}
This means that after a finite number of steps we get 
$$S=\bar{S}_t \oplus b_t \F_{q^{2^t}}\oplus b_{t-1} \F_{q^{2^{t-1}}}\oplus \cdots \oplus b_1 \F_{q^2} \oplus b \fq,$$
where $\bar{S}_t=a_t \langle 1, \lmb_t,\dots,\lmb_t^{h_t-1}\rangle_{\F_{q^{2t}}}$ for some $\lmb_t\in\fqn$ and integer $h_t\geq 0$ and  for any $j\in\{0,\dots, t-1\}$ the code $\mathrm{Orb}(\Bar{S}_t\oplus b_t\F_{q^{2^{t}}}\oplus\dots\oplus b_{j}\F_{q^{2^{j}}})$ is an FWS code in $\mathcal{P}_{q^{2^j}}(n)$. 
In particular,
\[
\Bar{a_t}^{-1}\Tilde{S}=\Bar{a_t}^{-1}(\Bar{S}_t\oplus b_t\F_{q^{2^{t}}} \oplus b_{t-1}\F_{q^{2^{t-1}}})= \langle 1, \lmb_t,\dots,\lmb_t^{h_t-1}\rangle_{\F_{q^{2t}}} \oplus \frac{b_t}{\Bar{a_t}}\F_{q^{2^{t}}} \oplus \frac{b_{t-1}}{\Bar{a_t}} \F_{q^{2^{t-1}}}
\]  
defines an FWS code in $\mathcal{P}_{q^{2^{t-1}}}(n)$,  and this is not possible because of Proposition  \ref{prop:pushup2}. So also Case 2.2) does not occur. This means that the only case that can occur is Case 1.1). Finally, Corollary \ref{cor:FWSq2} implies Theorem \ref{thm:casopiccoloFWS}
\end{proof}

\subsection{Proof of the Main Theorem}\label{subs:maintheoremfinal}

In this section we prove the main theorem of this paper, that is Theorem \ref{thm:maintheorem}.

\begin{proof}
Since $\cC=\mathrm{Orb}(S)$ is an FWS code, we can apply Theorem \ref{thm:class} to $\C$. In Case i) of Theorem \ref{thm:class}, we get that $S$, up to cyclic shift,  is as in $(1)$ of the statement and $m(\cC)=\dim_{\fq}(\fq (\lambda))$. From Corollary \ref{cor:existFWSpolbas} we get the complete characterization of subspaces with a polynomial basis defining FWS codes.  In Case ii) of Theorem \ref{thm:class}, using Theorems \ref{thm:noFWScode} and \ref{thm:casopiccoloFWS}, we get that $S$, up to cyclic shift,  is as in $(2)$ of Theorem \ref{thm:maintheorem} and by Corollary \ref{cor:FWSq2} we completely determine the case in which such subspaces define FWS codes. 

\end{proof}

\section{Concluding remarks}

In this paper we have provided a characterization result for one-orbit cyclic subspace codes without zeroes in their weight distribution, which we called \emph{full weight spectrum codes}. This result may be also stated in a purely geometric way. Indeed, our classification result allows us to determine the orbits, under the action of a Singer group, of projective subspaces intersecting the orbit's elements in all possible dimensions. Precisely, we have the following result. 

\begin{theorem}
    \label{th:projective subspace}
    Let $\Sigma_{n-1}=\mathrm{PG}(n-1,q)$ be the $(n-1)$-dimensional projective space over $\fq$, let $G$ be a Singer cycle of $\Sigma_{n-1}$ and let $\mathcal{G}_q(n-1, k-1)$ be  the Grassmaniann of the $(k-1)$-dimensional subspaces  
of $\Sigma_{n-1}$.  Let $\Sigma_{k-1}$ be a $(k-1)$-dimensional subspace  of $\Sigma_{n-1}$ and let $\mathcal{I}_{\mathrm{Orb}_G(\Sigma_{k-1})}=\{\Sigma_{k-1} \cap \varphi (\Sigma_{k-1})\, :\, \varphi \in G \}$ be the set of all subspaces obtained intersecting two distinct elements of $\mathrm{Orb}_G(\Sigma_{k-1})$.
Then 
$$ \mathcal{I}_{\mathrm{Orb}_G(\Sigma_{k-1})} \cap \mathcal{G}_q(n-1,d-1) \neq \emptyset$$
for every $d\in \{0,\dots, k-1\}$ if, and only if,  $\Sigma_{k-1}$, up to the action of $\mathrm{PGL}(n,q)$, is defined by one of the vector subspaces defined in Theorem \ref{thm:maintheorem}. 
\end{theorem}

Theorem \ref{thm:maintheorem} may be also used in other contexts such as in the theory of rank-metric codes.
Let $\C$ be a rank-metric code in $\F_{q^n}^{2k}$ with $\dim_{\F_{q^n}}(\C)=2$ and having a generator matrix $G$ whose column span (over $\fq$) is of the form $S\times S$, for some $\fq$-subspace $S$ of $\fqn$ with dimension $k$. Then, using \cite[Theorem 2]{Randra}, we have that
\[ w(xG)=n-\dim_{\fq}((S\times S)\cap x^\perp). \]
Therefore, $\C$ has no zeroes in its weight distribution (\emph{FWS} rank-metric code) if and only if $\mathrm{Orb}(S)$ is an FWS cyclic subspace code.
Therefore, Theorem \ref{thm:maintheorem} provides a classification result for FWS rank-metric codes with the above assumptions.
There are still some problems that naturally arise in this paper:

\begin{itemize}
    \item It would be nice  to determine the weight distribution of the codes in (2) of Theorem \ref{thm:maintheorem}.
    We already know that the weight distributions of the codes in (1) and those in (2) of Theorem \ref{thm:maintheorem} are different.
    Computational results show that the weight distribution does not depend on $\lambda$.
    \item Is it possible to determine the equivalence classes of the codes in Theorem \ref{thm:maintheorem} under the action of linear isometries?
    \item A natural generalization of FWS codes can be the following: let $r$ be a natural number. We say that a code is $r$-FWS code if the last $r$ entries of the weight distributions are zero and all the others are nonzero. In particular, $0$-FWS codes correspond to FWS codes. Is it possible to characterize $r$-FWS codes similarly to what has been done for FWS codes?
\end{itemize}

\section*{Acknowledgements}

The research was partially supported by the project COMBINE of ``VALERE: VAnviteLli pEr la RicErca" of the University of Campania ``Luigi Vanvitelli'' and  by the INdAM - GNSAGA Project \emph{Tensors over finite fields and their applications}, number E53C23001670001.

\section*{Declaration of competing interest}

The authors declare that they have no known competing financial interests or personal relationships that could have appeared to influence the work reported in this paper.

\section*{Data availability}
No data was used for the research described in the article.

Chiara Castello, Olga Polverino and Ferdinando Zullo\\
Dipartimento di Matematica e Fisica,\\
Universit\`a degli Studi della Campania ``Luigi Vanvitelli'',\\
I--\,81100 Caserta, Italy\\
{{\em \{chiara.castello,olga.polverino,ferdinando.zullo\}@unicampania.it}}

\begin{thebibliography}{8}

\bibitem{1Ald}
{\sc T. Alderson:} A note on full weight spectrum codes, \emph{Transactions on Combinatorics} \textbf{8(3)} (2019), 15--22. 

\bibitem{MWS}
{\sc T. Alderson and A. Neri:}
Maximum weight spectrum codes,
\emph{Advances in Mathematics of Communications} {\bf 13(1)} (2019), 101--119.

\bibitem{BSZ2015}
{\sc C. Bachoc, O. Serra and G. Z\'emor:}
An analogue of Vosper’s theorem for extension fields,
\emph{Mathematical Proceedings of the Cambridge Philosophical Society} {\bf 163(3)} (2017), 423--452.

\bibitem{BSZ2018}
{\sc C. Bachoc, O. Serra and G. Z\'emor:}
Revisiting Kneser's theorem for field extensions,
\emph{Combinatorica} {\bf 38(4)} (2018), 759--777.



\bibitem{countcoppol}
{\sc S. Corteel, C. Savage, H. Wilf, D. Zeilberger:}
A pentagonal number Sieve, \emph{Journal Combinatorial Theory, Series A} {\bf 82(2)} (1998), 186--192



\bibitem{Del4par}
{\sc Ph. Delsarte:}
Four fundamental parameters of a code and their combinatorial significance, \emph{Information and Control} {\bf 23(5)} (1973), 407--438.

\bibitem{Eliahou2009}
{\sc S. Eliahou, C. Lecouvey:} 
On linear versions of some additive theorems, 
\emph{Linear and Multilinear Algebra} {\bf 57} (2009), 759--775.



\bibitem{Etzion}
{\sc T. Etzion and A. Vardy:} 
Error-correcting codes in projective space,
{\em IEEE Transactions on Information Theory}  {\bf 57(2)} (2011), 1165--1173.




\bibitem{heideweight}
{\sc H. Gluesing-Luerssen and H. Lehmann:}
Distance distributions of cyclic orbit codes, \emph{Designs, Codes and Cryptography} {\bf 89} (2021), 447--470.

\bibitem{Heideequiv}
{\sc H. Gluesing-Luerssen and H. Lehmann:}
Automorphism groups and isometries for cyclic orbit codes, 
\emph{Advances in Mathematics of Communications} {\bf 17(1)} (2023), 119--138.


\bibitem{HouLeungXiang2002}
{\sc X. Hou, K.H. Leung and Q. Xiang:}
A generalization of an addition theorem of Kneser,
\emph{Journal of Number Theory} {\bf 97} (2002), 1--9. 

\bibitem{singerconjugate}
{\sc Huppert B.:} \emph{Endliche Gruppen}, Springer, Berlin, Heidelberg, New York, 1967.

\bibitem{JVdV}
{\sc D. Jena and G. Van de Voorde:}
On linear sets of minimum size,
\emph{Discrete Mathematics} {\bf 344} (2021), 112230.

\bibitem{KoetterK}
{\sc R. Koetter and F. R. Kschischang:} 
Coding for errors and erasures in random network coding, {\em IEEE Transactions on Information Theory} {\bf 54} (2008), 3579--3591.






\bibitem{NPSZminsize}
{\sc V. Napolitano, O. Polverino, P. Santonastaso and F. Zullo:}
Classifications and constructions of minimum size linear sets and critical pairs,
\emph{Finite Fields and Their Applications} {\bf 92} (2023), 102280.






\bibitem{Randra}
{\sc T.H. Randrianarisoa:} 
A geometric approach to rank metric codes and a classification of constant weight codes, 
\emph{Designs, Codes and Cryptography} {\bf 88} (2020), 1331--1348,



\bibitem{Shi}
{\sc M. Shi, H. Zhu, P. Sol\'e and G.D. Cohen:}
How many weights can a linear code have?,
\emph{Designs, Codes and Cryptography} {\bf 87} (2019), 87--95.

\bibitem{Trautmann2}
{\sc A.-L. Trautmann:} 
Isometry and automorphisms of constant dimension codes, 
\emph{Advances in Mathematics of Communications} {\bf 7} (2013), 147--160.


\bibitem{Vosper}
{\sc G. Vosper:}
The critical pairs of subsets of a group of prime order,
\emph{Journal of London Mathematical Society} {\bf 31} (1956), 200--205.


\end{thebibliography}
\end{document}